\documentclass[%
 reprint,
 superscriptaddress,
 nofootinbib,
 amsmath,amssymb,
 aps,
 pra
]{revtex4-2}

\usepackage{graphicx}
\usepackage{dcolumn}
\usepackage{bm}
\usepackage[colorlinks,linkcolor=blue,urlcolor=blue, citecolor=blue]{hyperref}
\usepackage{lipsum}
\usepackage{dsfont}
\usepackage{physics}
\usepackage{enumitem}
\usepackage[most]{tcolorbox}
\usepackage[linesnumbered,ruled,vlined]{algorithm2e}
\usepackage{needspace}
\usepackage{braket}
\usepackage{amsthm} 
\usepackage{thmtools} 
\usepackage{thm-restate}
\usepackage{tabularx}

\declaretheorem{corollary,lemma,proposition,definition}

\usepackage{amsmath, amssymb, amsfonts, amsthm,mathtools}

\let\oldaddcontentsline\addcontentsline
\renewcommand{\addcontentsline}[3]{}

\renewcommand{\tr}{\mathrm{Tr}}

\begin{document}

\preprint{APS/123-QED}

\title{Autonomous Quantum Processing Unit:\\ 
\texttt{An Autonomous Thermal Computing Machine \& its Physical Limitations}}

\author{Florian Meier}
\email[]{florianmeier256@gmail.com}
\affiliation{Atominstitut, TU Wien, 1020 Vienna, Austria}
\affiliation{Vienna Center for Quantum Science and Technology, TU Wien, 1020 Vienna, Austria}

\author{Marcus Huber}
\email[]{marcus.huber@tuwien.ac.at}
\affiliation{Atominstitut, TU Wien, 1020 Vienna, Austria}
\affiliation{Vienna Center for Quantum Science and Technology, TU Wien, 1020 Vienna, Austria}
\affiliation{Institute for Quantum Optics and Quantum Information (IQOQI), Austrian Academy of Sciences,
Boltzmanngasse 3, 1090 Vienna, Austria}
\author{Paul Erker}
\email[]{paul.erker@tuwien.ac.at}
\affiliation{Atominstitut, TU Wien, 1020 Vienna, Austria}
\affiliation{Vienna Center for Quantum Science and Technology, TU Wien, 1020 Vienna, Austria}
\affiliation{Institute for Quantum Optics and Quantum Information (IQOQI), Austrian Academy of Sciences,
Boltzmanngasse 3, 1090 Vienna, Austria}
\author{Jake Xuereb}
\email{jake.xuereb@tuwien.ac.at}
\thanks{corresponding author.}
\affiliation{Atominstitut, TU Wien, 1020 Vienna, Austria}
\affiliation{Vienna Center for Quantum Science and Technology, TU Wien, 1020 Vienna, Austria}

\date{\today}

\begin{abstract}
Computation is an input-output process, where a program encoding a problem to be solved is inserted into a machine that outputs a solution. Quantum computation conventionally relies on classical, external control outside the quantum computer to execute a program, obscuring computational and thermodynamic resources required. To understand the fundamental limits of computation, however, it is pivotal to work with a fully self-contained description of a quantum computation modeling the resources on the same footing as the computation itself. By developing a framework that we dub the \textit{autonomous Quantum Processing Unit} (aQPU) we model quantum computation in the framework of autonomous thermal machines. Consisting of an internal quantum timekeeping mechanism, instruction register and memory system the aQPU allows investigating relationships between thermodynamic cost, complexity, speed and fidelity of a desired quantum computation.
\end{abstract}

\maketitle

\begin{figure*}[ht]
    \centering
    \includegraphics[width=\textwidth]{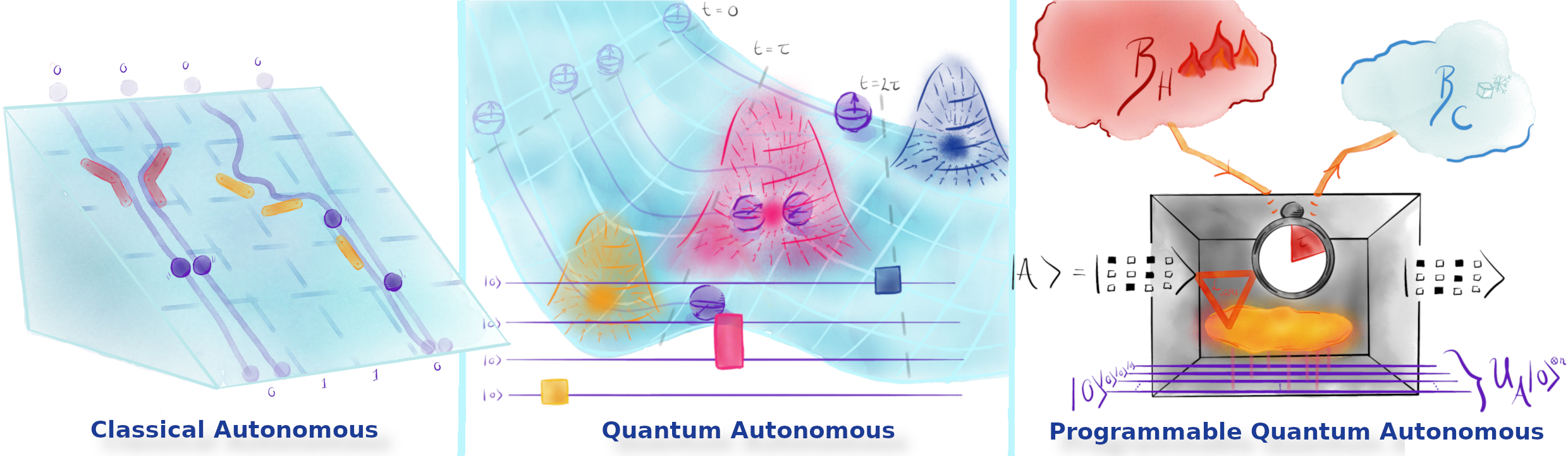}
    \caption{In these three panels we give artistic impressions illustrating the setting for our proposed model, conveying what we understand by \textit{autonomous}, i.e., the absence of time-dependent classical control mediated by macroscopic fields. In particular in the left panel we give an illustration of an autonomous model of classical computation, the billiard-ball computer. In the central panel we visualize an illustrative example described below of an model of quantum computation where spin-$1/2$ particles progress down a potential landscape with several interaction regions which altering the state of these particles. In the right panel we visualise the proposed autonomous quantum processing unit, an autonomous thermal machine capable of universal quantum computation -- featuring an autonomous thermal clock driven by out-of-equilibrium baths, a tick register, computational register and instruction register.}
    \label{fig:concept_art}
\end{figure*}

The earliest conceptions of computation including Babbage's Analytical Engine~\cite{babbage} and the Turing Machine~\cite{turing} envisaged computers as objects which receive a mathematical problem as an input and output its solution, culminating in the modern day realization of these ideas in silicon chips. 
Quantum physics has ushered in a new paradigm for probabilistic computation also introduced in the form of an input-output device by Benioff~\cite{Benioff1980,Benioff1982}, Deutsch~\cite{deutsch} and Bernstein \& Vazirani~\cite{vazirnai}, the quantum Turing Machine.
Separately, Feynman~\cite{Feynman2023} and Kitaev~\cite{KSV02} also envisaged a method for encoding a quantum computation into a Hamiltonian (also known as the circuit-to-Hamiltonian mapping~\cite{cirac_2008,Bausch2016TheCO,Jordan2017,Bausch_2018}) which maintains the input-output and autonomy features that computation was originally formulated in.
Both are successful conceptual frameworks for understanding complexity and computability in the context of quantum mechanics, but have yet to be thermodynamically grounded in an open system analysis.

Current quantum computer architectures use classical, external control to execute a computation~\cite{Saffman2010,Kjaergaard2020a}. 
A sequence of computational gates is applied to an input quantum state prior to read-out at the end. The choice of gate and order in which they are implemented requires an external instruction that is still classical. Additionally, each gate is implemented as a unitary operation via judiciously chosen Hamiltonians that drive the system at a precise time and energy, guided by a classical control architecture.

Some practical costs required for quantum computation are known: cooling and isolating physical systems allowing them to exhibit quantum behavior, and even some fundamental costs, e.g., Landauer's cost of erasing bits of information~\cite{landauer}.
Despite this, we know very little of the fundamental energetic costs of quantum computing.
Mainly due to the fact that no complete, self-contained model of \textit{autonomous} quantum computation exists.
In this context, \textit{autonomy} is to be understood in analogy to the billiard-ball model for a classical computer introduced in the seminal work ``Conservative Logic'' by Fredkin and Toffoli~\cite{Fredkin2002}: a quantum computer, which, after initial preparation, evolves without external intervention as illustrated in Fig.~\ref{fig:concept_art}.

In order to model computation in a truly autonomous and self-contained way, two conceptual challenges need to be overcome: 
Both (a) the instructions, and (b) the control to execute those instructions need to be described as part of the quantum computer’s model without referring to external agents.
In such an approach, energetic resources are treated within the model explicitly, enabling a first complete study of the cost of computation within a minimal, autonomous quantum framework.

To achieve (a), we model the quantum computer as a machine, where all computational steps in a program are timed and executed by an internal quantum clock, acting as a proxy for the physics of precise control in this model.
This program, labeled by $\mathcal A$, is encoded into a quantum state and initially fed into the machine.
After the computation finishes, the final state is output and ideally approximates the state as defined by a perfect computation according to the program.
We find that the more accurate the clock, i.e., the more precise the control, the higher the fidelity $\mathcal {F_A}$ with which an autonomous machine can approximate the desired state of the program,
\begin{align}
    \mathcal{F_A} = 1 - O\left(\frac{L}{N}\phi^2_\text{max}\right). \label{eq:main}
\end{align}
Here, $L$ is the length of the program as measured by the number of elementary gates to be executed,
and $N$ is the clock accuracy as defined in~\cite{Erker2017,Woods2021,Woods2022}.
Clock accuracy can be understood as the average number of times the clock ticks until it goes wrong by one tick.
The parameter $\phi_{\rm max}$ depends on the specific Hamiltonians required to execute the program, and for conventional universal gate sets like the Clifford+$T$ set~\cite{Nielsen2010}, the constant is of order $O(1)$.
Since higher clock accuracy $N$ is generally linked to higher entropy dissipation~\cite{Milburn2020,Erker2017,Barato2016}, this relation further shows how high computational fidelity is linked to high dissipation, resolving (b)---precise quantum computation has a thermodynamic cost.

Naturally, with the advent of quantum computing the question of its inherent energy consumption has been investigated from different angles~\cite{Auffeves2022,FellousAsiani2023,Jaschke2023,Meier2023a}.
One of the challenges in this respect is the fact that there are many energetic factors to a quantum computation in the circuit model, obscuring an understanding of the true fundamental cost of quantum information processing from a thermodynamic perspective.
A particular issue also raised by the analysis in~\cite{FellousAsiani2023} is that typically employed time-dependent Hamiltonians involve large classical control costs (such as lasers or large magnetic fields) that contribute far more to the energetics of the quantum computation than that arising from the information processing at the quantum level.
We make use of the framework of quantum thermal machines~\cite{Goold2016,Mitchison_2019,skrzypczyk_2010} and recent results from the field of quantum clocks~\cite{Erker2014,Erker2017,Meier2023,Xuereb2023} to introduce a model of autonomous quantum computation in the language of open quantum systems where the non-equilibrium thermodynamic cost such as entropy production can be fully accounted for~\cite{petruccione_breuer,Hofer2017,Carmichael2002,Walls2008,landi_paternostro_rev}.
This introduces the major technical challenge solved in this work. The question of how the output of an autonomous quantum clock could be used to control another quantum system has only been attempted in single operations and not arbitrary computations~\cite{Malabarba_2015,Woods_2023}. By addressing this gap we show that even quantum computation which appears at first glance reversible due unitarity, incurs a fundamental dissipative cost independent of measurement, due to the quality of its control.

\section*{Model}

\begin{table}[h]
\def\arraystretch{1.2}
\begin{tabularx}{\columnwidth}{l|l}
\hline
System                        & Subscript \\ \hline\hline
Memory Register \hspace{4cm}  & $M$   \\ \hline
Instruction Register          & $I$   \\ \hline
Clockwork                     & $C$   \\ \hline
Tick Register                 & $T$   \\ \hline
\end{tabularx}
\caption{Throughout this section the following notational shorthand is introduced for denoting different parts of the aQPU.\label{tab:subscripts}}
\end{table}

We begin by introducing the four components which constitute the aQPU:
the memory register where the computation is carried out,
the instruction register containing the program,
the tick register controlling which operation is being executed and finally the clockwork timing everything (summarized in Tab.~\ref{tab:subscripts}).

\paragraph*{Memory register.}
We are interested in the case where the aQPU can apply gates from a set $\mathcal V$ consisting of a finite number of unitary operations labeled $V_{M,1},\dots,V_{M,K}$, the subscript $M$ denoting the memory register upon which they act.
In the special case where the memory system is a register of qubits, a particularly relevant choice for $\mathcal V$ is one where the gates are drawn from a universal gate set like for example the Clifford+$T$ gateset.
Due to the Solovay-Kitaev Theorem~\cite{Kit97,KSV02,Nielsen2010,dawson2006solovay}, any target quantum state in the Hilbert space can be reached with an error that decreases as a stretched exponential in the length of a specific sequence (further details can be found in Sec.~\ref{sec:speed_vs_fidelity}).
Equivalently, the gates in $\mathcal V$ can be written as the exponential of a Hamiltonian evolved for some time $\tau$.
If we suppose that
\begin{align}
\label{eq:V_Mk_def}
    V_{M,k}=e^{-iH_{M,k}\tau}   
\end{align}
for some Hamiltonian $H_{M,k}$, the set $\{H_{M,1},\dots,H_{M,K}\}$ is sufficient to generate $\mathcal V$. 
For a given aQPU both the duration $\tau$ and the energy scale of the Hamiltonians are assumed to be non-tunable quantities that are fixed by the dynamics of the physical interactions \textit{a priori}, such that they do not depend on the algorithm an agent wishes to carry out.

\paragraph*{Instruction Register.}
A quantum computation of length $L$ can be defined by a sequence of gates $V_{M,a_0},\dots,V_{M,a_{L-1}}$ from $\mathcal V$ applied to a well-defined initial state of the memory register $\ket{0}_M$. By the  $L$-tuple   $\mathcal A = (a_0,\dots,a_{L-1})$, we denote the program $\mathcal A$.
Applying all gates in $\mathcal A$ in sequence then defines the unitary of the program.
\begin{align}
\label{eq:VA_def}
    V_{\mathcal A} := V_{M,a_{L-1}}\cdots  V_{M,a_1}V_{M,a_0}.
\end{align}
For the self-contained operation of the aQPU, the program must be encoded in a quantum system, which we call the instruction register $I$.
One way to achieve this is to enumerate each of the $K$ gates using the  elements of the computational basis  for a $(K+1)$-dimensional Hilbert space spanned by the states $\ket{0},\ket{1},\dots,\ket{K}$.
The states $\ket{1},\dots,\ket{K}$ then encode the unitaries $V_{M,1},\dots,V_{M,K}$ respectively, and the state $\ket{0}$ encodes an idle gate, i.e., the identity gate that does nothing.
The full program $\mathcal A$ can then be written as a \textit{punch card state},
\begin{align}
    \ket{\mathcal A}_I = \ket{a_0}_{I_0}\ket{a_1}_{I_1}\cdots  \ket{a_{L-1}}_{I_{L-1}}\ket{0}_{I_{L}} \cdots \ket{0}_{I_m},
\end{align}

where $m$ denotes the program's maximum length.
This provides an explicit scheme by which instructions are encoded at the quantum level~\cite{kjaergaard2020programming} in a quantum instruction register.

\paragraph*{Tick register and clockwork.}
Autonomously switching between the instructions encoded in the punch card state can be achieved by including a clock in the aQPU.
This \textit{master clock} should ideally switch between the instructions after a time $\tau$ such that the desired gate is implemented following Eq.~\eqref{eq:V_Mk_def}.
Quantum mechanical models for clocks that generate a discrete periodic signal have been proposed for example in the works~\cite{Erker2017,Woods2019,Woods2021,Woods2022,Woods_2023,Dost2023,Silva2023,Meier2024b,peres}.
For controlling the sequential execution of the program, a tick register $T$ is required to count the ticks of the clock.
We use the notation $\ket{0}_T,\ket{1}_T,\dots,\ket{m}_T$ to denote these counter states, with a maximal state $\ket{m}_T$ to ensure the computation stops after the program has been executed.
A \textit{clockwork} $C$ then provides a time reference such that the jumps $\ket{n}_T\rightarrow\ket{n+1}_T$ ideally occur every time an interval of duration $\tau$ elapsed.
The evolution of such a clock can be described by a Lindblad master equation of the form $\dot\rho_{CT} = \mathcal L_C \rho_{CT} + \mathcal D_J[\rho_{CT}]$~\cite{Silva2023}.
The Lindblad operator $\mathcal L_C$ describes the open system's evolution of the underlying clockwork and generally comprises coherent and dissipative contributions.
To use the master equation description, weak coupling to the environments has to be assumed which is often compatible with how such clocks are defined~\cite{Hofer2017,Erker2017,Dost2023} (details in the Appendix~\ref{appendix:applicability_ME}).
The second contribution in the clock evolution comes from the tick dissipator $\mathcal D_{J}[\rho]=J\rho J^\dagger - 0.5\{J^\dagger J,\rho\}$ generating the ticks of the clock.
The jump operator is usually of the form $J=J_C\otimes \sum_{n=0}^{m-1}\ketbra{n+1}{n}_T$, with $J_C$ some jump acting only on the clockwork.
Each jump with $J$ then shifts the tick register by one unit.

\paragraph*{Universal interaction Hamiltonian.}
With the parts constituting the aQPU introduced, we now present the interaction which brings them together to act as a computational machine. To enact the correct Hamiltonians in the right sequence on the memory register corresponding to the programmed gates in $\mathcal A$, we propose a three-body interaction Hamiltonian.
This Hamiltonian connects the tick register with the instruction and memory register,
\begin{align}
    \label{eq:H_int}
    H_\text{int} = \hspace{-0.4cm}\sum_{\substack{0\leq n \leq m-1 \\ 1\leq k \leq K}} \hspace{-0.2cm}\left(\ketbra{n}{n}_{T} \otimes \mathds1_{I_{m(\neq n)}} \otimes \ketbra{k}{k}_{I_n} \otimes H_{M,k}\right).
\end{align}

The choice of a three-body interaction is for convenience of the mathematical representation. 
In practice, to connect the tick register with both the instruction and memory register, no fundamental three-body interactions are required.
Instead, several two-body interactions can in principle be used to give rise to an effective three-body interaction---a mechanism that has recently been demonstrated experimentally~\cite{aamir2023thermally}.

Conceptually, one can think of this the interaction in Eq.~\eqref{eq:H_int} as enacting the Hamiltonian $H_{M,k}$ on the memory system conditioned on the clock counting $n$ ticks and corresponding $n$th program step reading $k$.
Note that there is no contribution from the $m$th tick $\ket{m}_{T}$ which ensures that once the the clock has ticked through all instruction steps of the punch card, the interaction on the memory system is turned off.
In the late time limit, all instruction states have then been carried out and the aQPU is left idling.

The full dynamics of the aQPU are thus a combination of the interaction Hamiltonian from Eq.~\eqref{eq:H_int}, the terms driving the clockwork and the ticks,
\begin{align}
\label{eq:L_aQPU}
    \mathcal L_\text{aQPU} = -i[H_{\rm int},\,\cdot\,] + \mathcal L_{C} + \mathcal D_{J}.
\end{align}
As such, time-evolution of $\rho(t)$ on the full Hilbert space is generated by the evolution equation $\dot \rho = \mathcal L_{\rm aQPU}\rho$.
These dynamics, in principle, allow approximating arbitrary programs $V_{\mathcal A}$ on the memory system, which is discussed more quantitatively in the following section.

\section*{Results}
How well the aQPU approximates the program unitary $V_{\mathcal A}$ depends on the perfection of its underlying control mechanisms, e.g., precision of interaction strength of control fields or timers. As a proxy for this, let's consider how the quality of the aQPU's clock mechanism impacts its performance. 
It is known from previous works~\cite{Aaberg2014,Ball2016,Xuereb2023} that non-ideal timing of unitary gates using a classical tick register leads to dephasing. 
One way to quantify the quality of a clock is by using the probability distribution of its ticks.
This distribution is denoted by $P[t]={\rm Pr}[T_1=t]$, where $T_1$ is the random variable describing the time between subsequent ticks.
More generally, $T_n$ describes the time between $n+1$ ticks.
For the present analysis, we consider the case where the time between subsequent ticks of the clock all have an identical distribution and are independent from each other, that is, they are i.i.d., which need not be the case in the most general model of a clock.
Then, the \textit{clock accuracy} defined as $N=\mu^2 / \sigma^2$, where $\mu$ is the average and $\sigma^2$ the variance of $P[t]$ can be used to quantify how sharp the clock's ticks are in time.
This notion has been successfully employed in a variety of works on clocks~\cite{Erker2017,Woods2019,Woods2021,Woods2022,Woods_2023,Dost2023,Silva2023,Meier2023}, and colloquially, $N$ is the average number of times the clock ticks until it goes wrong by one average tick time.

\paragraph*{Idealized case.}
As a preliminary calculation, let us verify that if the clock underlying the aQPU is perfect, any encoded program $\mathcal A$ is executed without errors.
An ideal clock in this sense is one whose tick distribution $P[t]=\delta(t-\tau)$ is deterministic, here, with mean tick time scaled to $\mu=\tau$.
The accuracy in this limit diverges $N\rightarrow\infty$.
At the same time, the evolution generated by $\mathcal L_{\rm aQPU}$ restricted to the memory system reduces to the sequence of gates from the program.
To achieve this, the aQPU is initialized in the state $\rho^{\rm init} = \rho_{C}^{\rm init}\otimes\ketbra{0}{0}_T\otimes\ketbra{\mathcal A}{\mathcal A}_I\otimes \rho_{M}^{\rm init}$, with the tick register counting $0$ ticks, and the instruction register encoding the desired program $\mathcal A$.
After evolving for a time $t>L\tau$, all unitaries encoded in the program have been executed perfectly,
\begin{align}
    \tr_{CTI}\left[e^{\mathcal L_{\rm aQPU}t}\rho^{\rm init}\right] =V_{\mathcal A} \rho_{M}^{\rm init} V_{\mathcal A}^\dagger.
\end{align}
Details for the proof can be found in the Appendix~\ref{sec:aQPU_universality}.
A more physically relevant case is when we deviate from the idealized scenario.
Here, due to the variance in the clock's tick distribution we find that the aQPU is in a mixed state of all the number of times that its clock could have ticked leading to different computational trajectories.
Generally, the overall state is thus of the form,
\begin{align}
\label{eq:rho(t)_ansatz}
    \rho(t) = \sum_{0\leq n\leq m} \rho_C^{(n)}(t)\otimes \ketbra{n}{n}_T\otimes\ketbra{\mathcal A}{\mathcal A}_I \otimes \rho_M^{(n)}(t),
\end{align}
that is, block-diagonal with respect to the tick number.
As a convention, we chose that $\tr[\rho_M^{(n)}]=1$ is normalized.
Thus, the weight of the $n$th tick is encoded in $\tr[ \rho_C^{(n)}(t)]={\rm Pr}[N(t)=n],$ which is the probability that the tick register reads $n$.

\paragraph*{Recursion relation.}
The task of solving the aQPU's evolution equations can be simplified by noting that the clock dynamics can be decoupled from the memory system.
For a given tick probability distribution, the memory's evolution is solvable using the ansatz in Eq.~\eqref{eq:rho(t)_ansatz}.
The resulting evolution equation becomes
\begin{align}
    \dot \rho_M^{(n)}(t) = -i\big[H_{M,a_n},\rho_M^{(n)}&(t)\big] +\nonumber\\ &p(t)\big(\rho_M^{(n-1)}(t)-\rho^{(n)}_M(t)\big),    \label{eq:dot_rho_T}
\end{align}
with $p(t)={\rm Pr}[T_n=t]/{\rm Pr}[N(t)=n]$, and $\rho_M^{(-1)}\equiv 0$.
Note that ${\rm Pr}[T_n=t]$ is the probability density that the time between $n$ ticks equals $t$.
These equations can be solved recursively, with an integral solution given in Appendix~\ref{sec:aQPU_universality}.

Having investigated the general case, let us now restrict to the case where the aQPU's clock is sufficiently well-behaved and accurate, but still imperfect.
This allows for a more direct, approximate form of the clock's memory state.
To arrive at this form we being by noting that while high accuracy implies small width $\sigma$ of the clock's tick distribution relative to the mean tick time, this in general does not constrain the higher moments of the tick distribution.
To ensure the higher moments do not conspire in a way that errors accumulate, the clock's distribution needs to be sufficiently bounded.
To quantify this more rigorously, we consider a family of clocks whose average tick time is equal $\mu=\tau$ but whose accuracy is unbounded $N\rightarrow\infty$.
This family of clocks is said to have \textit{exponentially concentrated} tick probabilities if there exist constants $\alpha,c>0$ such that for all times $t$ the tick probability is bounded by an exponential, ${\rm Pr}[|T-\tau|\leq t]\leq \alpha e^{-c\sqrt{N}t}$.

We wish to consider the final result of the computation, i.e., the aQPU state at some time later than $L\tau$.
For example, at twice the expected finishing time $t_f=2L\tau$, we can be sufficiently sure that the computation has concluded.
This time is of the same order as the program's expected runtime $L\tau$, but sufficiently longer than the time $\tau$ between two ticks.
For an aQPU with exponentially concentrated ticks, the memory's state after $t_f$ is exponentially well approximated by the steady-state,
\begin{align}
\label{eq:rho_M_geqT_approx}
    \rho_M(t)\Big|_{t\geq t_f} = \rho_M^{(L)} + O(e^{-cL\sqrt{N}/2}).
\end{align}
Here, $\rho_M^{(L)}$ is a time-independent state given by the idealized recursion relation reminiscent of the expression in Eq.~(2) of~\cite{Xuereb2023},
\begin{align}
\label{eq:rho_M_ideal_recursion}
    \rho_M^{(n+1)}=\int dt P[t] V_{M,a_n}(t) \rho_M^{(n)} V_{M,a_n}^\dagger (t),    
\end{align}
where $V_{M,a}(t) = e^{-iH_{M,a} t}$ is the Hamiltonian corresponding to the instruction $a$ evolved for some time $t$.
Given a family of clocks with exponentially concentrated ticks, it turns out that the recursion relation simplifies to
\begin{align}
\label{eq:rho_M_recursion_approx}
    \rho_M^{(n+1)} = V_{M,a_n} \rho_M^{(n)} V_{M,a_n}^\dagger +  O\left(\frac{\tau^2\|H_{M,a_n}\|^2}{N}\right),
\end{align}
in the limit of large clock accuracy $N\gg\tau^2 \|H_{M,a_n}\|^2$ (for details see Appendix \ref{sec:error_propagation_nonideal_clocks}, specifically proof of Proposition \ref{prop:clockchannel}). 
Note that $V_{M,a_n}$ is the desired unitary as in Eq.~\eqref{eq:V_Mk_def}, corresponding to the $n$th step $a_n$ in the program $\mathcal A$.
The remainder term becomes vanishingly small in the limit of high accuracy.
When considering the recursion relation in Eq.~\eqref{eq:rho_M_ideal_recursion}, however, errors accumulate over the course of the computation.

To quantify who well the aQPU approximates the desired computation, we make the assumption that the quantum computer starts in a well-defined initial state $\ket{0}_M$ and the desired final state is denoted by $\ket{\Psi({\mathcal A})}_M = V_{\mathcal A}\ket{0}_M$.
We can then calculate the fidelity of the desired state $\ket{\Psi(\mathcal A)}_M$ with the actual state $\rho_M(t_f)$ of the memory system after the expected time it takes for the program $\mathcal A$ to finish, given by $\mathcal F_{\mathcal A}:=\langle \Psi(\mathcal A)|\rho_M(t_f)|\Psi(\mathcal A)\rangle$.
Note that the state $\rho_M(t)$ does not vary strongly for values $t\geq t_f$ because by construction, the $(L+1)$st instruction is idle. 
By concatenating the recursion relation in Eq.~\eqref{eq:rho_M_recursion_approx} for all the instructions and using Eq.~\eqref{eq:rho_M_geqT_approx}, we find that $\rho_M(t_f)=\ketbra{\Psi(\mathcal A)}{\Psi(\mathcal A)}+ O(L \phi_{\rm max}^2 / N)$.
The parameter $\phi_{\rm max} = \tau \max_a \|H_{M,a}\|$ is the maximum  prefactor (of the additive term)  on the right-hand side of Eq.~\eqref{eq:rho_M_recursion_approx} over all instructions, and can be understood as a generalized angle about which the aQPU rotates the input state throughout the computation.
Thus, the deviations from the desired state for non-ideal clocks depend inverse linearly on the clock accuracy $N$, with a prefactor depending on the length $L$ the program and $\phi_{\rm max}$.
We can rephrase this as a statement about the fidelity $\mathcal F_{\mathcal A}$:
For an aQPU with master clock producing exponentially concentrated i.i.d.\ ticks at accuracy $N\gg L,L\phi_\mathrm{max}^2,$ the final program fidelity $\mathcal F_{\mathcal A}$ for any program $\mathcal A$ of length $L$ is at least
\begin{align}
    \mathcal F_{\mathcal A} = 1 - O\left(\frac{L}{N}\phi_\mathrm{max}^2\right). \label{eq:prec}
\end{align}
A detailed proof of this statement can be found in Appendix~\ref{sec:error_propagation_nonideal_clocks}.

\paragraph*{Thermodynamic cost.}
The self-contained description of the aQPU as an open quantum system makes it possible to rigorously quantify the thermodynamic cost of the computation beyond established approaches~\cite{bennett1982thermodynamics,Auffeves2022,manzano,wolpert_1} using entropy production.
The three contributions are: (1) the preparation of the aQPU's initial state, (2) entropy production during the evolution of the aQPU, and (3) measurement or other appropriate readout of the final state.
In the following, we focus on the contribution from point~(2).
The initialization from point~(1) simplifies to the problem of initializing tick, instruction and memory register.
This can be formulated for example as a quantum cooling task for preparation of pure states for which bounds on the entropic cost have been studied previously~\cite{Reeb2014,Taranto2023,Xuereb2023a} (details in Appendix~\ref{sec:termodynamic_cost_of_computation}).
As for point~(3), we note that characterizing the fundamental thermodynamic cost of quantum measurements is widely considered to be an open problem~\cite{Granger2011,Deffner2016,Guryanova2020} not unique to this work.
We leave a detailed analysis of the measurement cost for future work, because it can always be considered as independent of the quantum algorithm.
The reason is that any choice of measurement basis can always be reduced to a universal basis choice preceded by an appropriate unitary transformation, whose cost is covered by our analysis for point (2).

For analyzing the contributions from point (2), we turn to the entropic contributions from the clock as a proxy for the dissipation inherent of precise control~\cite{Barato2015,Horowitz2020}.
Note that the evolution due to the interaction Hamiltonian is reversible and so does not come at an entropic cost.
The clock on the other hand is necessarily an open system.
For environment that satisfy local detailed balance~\cite{landi_paternostro_rev,Spohn1978}, the interaction between the clock and the environment has been shown to come at a thermodynamic cost in terms of entropy~\cite{Erker2017,Schwarzhans2021,Dost2023}.
A general principle is that the more accurate a clock is, the higher the entropic cost per tick---implying that higher fidelity for a computation as per Eq.~\eqref{eq:prec} comes at a higher entropic cost.
For clocks with a fixed entropic cost per tick, there exists a model-dependent functional relationship $N= f(\Sigma_{\rm tick})$ between the accuracy $N$ and the entropy per tick $\Sigma_{\rm tick}$.
For example, for dissipative clock models like e.g.~\cite{Erker2017}, the function is $f(x) \propto x$.
For quantum clock models, a less stringent bound usually exists, for example $f(x) \propto x^2$ as in~\cite{Dost2023} or $f(x) = e^{\Omega(x)}$ as in~\cite{Meier2024b}, where $\Omega$ is the Knuth-notation for an asymptotic lower bound~\cite{Knuth1976}.
Combining these bounds with Eq.~\eqref{eq:prec} gives a bound on the fidelity in terms of the entropy dissipation, $\mathcal F_{\mathcal A} = 1 -  O(L \phi_{\rm max}^2/f(\Sigma_{\rm tick}))$.
The thermodynamic uncertainty relations provide a relationship between precision and entropy production for more general settings~\cite{Barato2015,Barato2016,Horowitz2020,Pal2021a}; they thereby prescribe the specific functional form of $f(x)$, which, in the case of dissipative classical clocks, reduces to the linear scaling $f(x)\propto x$.

\paragraph*{Compilation of gates from a universal set.}
In the special case where the gate set of the aQPU is universal, any desired unitary $U$ to be carried out by the aQPU can be approximated arbitrarily well by a product of gates $V_{\mathcal A}=V_{M,a_{L-1}}\cdots V_{M,a_0}$ from the universal gate set, due to the Solovay-Kitaev Theorem~\cite{Kit97, dawson2006solovay}.
The longer the program $\mathcal A$ used to compile the desired unitary, the smaller the error of the approximation. When considering single qubit gates, so $U\in {\rm SU}(2)$, the error $\varepsilon=\|U-V_{\mathcal A}\|_\infty$ decays as a stretched exponential $\varepsilon=\exp(-{\Omega(L^{1/c})})$ in the length $L$ (number of gates in $\mathcal A$) of the approximation, where $c>0$ is some constant~\cite{dawson2006solovay}.
However, according to Eq.~\eqref{eq:prec}, the longer the program on the aQPU, the more the timekeeping error accumulates, competing with the decreasing error due to the Solovay-Kitaev Theorem, and giving rise to a sweet-spot for the length $L$ of the compilation.
Overall, the error between the desired state $U\ket{0}_M$ and the approximation $\rho_M(t_f)$ on the memory register is then given by the two contributions (details in Appendix~\ref{sec:speed_vs_fidelity}),
\begin{align}
    \hspace{-0.3cm}\left\|U\ketbra{0}{0}_M U^\dagger \!-\! \rho_M(t_f)\right\|_1
    \!\leq\! e^{-\Omega\left(L^{1/c}\right)} \!+\!  {O}\left(\frac{L}{N}\phi_\mathrm{max}^2\!\right)\!. \label{eq:speed}
\end{align}
\section*{Experimental proposal}
For a physical realization, one of the main challenges is coupling the tick register with the instruction and memory register through the interaction Hamiltonian.
A basic building block exhibiting the desired properties can be constructed already with two qubits using a dispersive shift interaction~\cite{Blais2021}.
In practice such a system could be achieved by coupling a spin-system to a resonator or optical cavity e.g. by using coupled Transmons to realize an artificial spin-system~\cite{Schuster2007,Benito2019,Blais2021,Pekola2015}. 

\paragraph*{Setup.}
In the setup, the first qubit serves as an abstraction of the control components, including the clockwork, ticking register and instruction register, while the second qubit is the memory system where the computation will occur.
The model can be described using the Pauli matrices $\sigma^i$, where $i=x,y,z$.
Conventionally, the Hamiltonians of the two qubits are given by $H_C=\frac{\omega_C}{2}\sigma_C^z$ for the control and $H_M=\frac{\omega_M}{2}\sigma_M^z$ for the memory, where $\omega_{C/M}$ are the respective qubit frequencies (units of $\hbar=1$).
Using the spontaneous decay at rate $\Gamma$ of the first qubit as a timekeeping mechanism, a gate on the memory can be timed with average duration $\mu=1/\Gamma$.
Since also $\sigma=1/\Gamma$, the accuracy is $N=1$~\cite{Erker2017} and we thus do not expect this toy model to yield practically relevant computational fidelity for the aQPU.
Still, it captures the key physical features required for autonomous control, and serves the basis for how more accurate quantum clock models based on larger Hilbert spaces can be involved in this setting.

\begin{figure}
    \centering
    \includegraphics[width=\columnwidth]{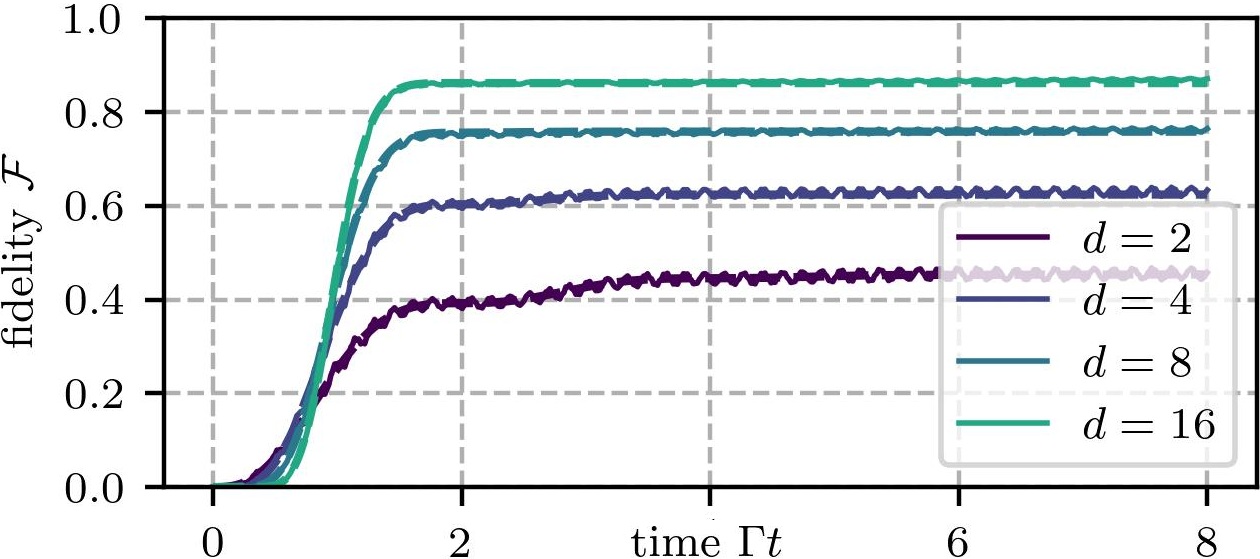}
    \caption{The fidelity between the desired and the actual state generated by the minimal aQPU implementation for different dimensions $d$ of the control system.
    The rotating wave approximation in (dashes line) is compared to an numerical integration (solid wiggly lines).
    The parameters $\Omega=\pi \Gamma$, $\omega = 100\Gamma$, $\chi= 10\Gamma$ have been chosen for the simulation.
    The simulation depicts a \textit{worst case} scenario, where the gate is supposed to rotate a state into an orthogonal one, maximally impacting the fidelity.
    Yet, we see that the fidelity grows with increasing size of the control system, and corresponding clock accuracy $N=d-1$ (c.f.~Refs.~\cite{Barato2016,Erker2017}).}
    \label{fig:RWA_implementation}
\end{figure}

By coupling the two qubits dispersively with a Hamiltonian of the form $H_{CM} = \chi \sigma_C^z \otimes \sigma_M^z$, the control qubit can dispersively shift the frequency of the memory qubit by $2\chi$~\cite{Blais2021}.
We use this shift to bring an external drive acting on the memory with  $H_{\rm drive}(t) = \Omega \cos(\kappa t) \sigma_M^x$ in- and out-of-resonance, with the corresponding resonance condition $\kappa = \omega_M+2\chi$.
The goal of this construction is to autonomously switch on and off the drive.
While the drive is in principle time-dependent, we will see that this time-dependence drops out in the effective description, and that the switching on and off of the drive does not rely on that time-dependence.

In the rotating frame relative to the Hamiltonian $H_0 = H_C+ H_M + H_{CM}$, and under the rotating wave approximation (RWA) where fast oscillating terms are neglected~\cite{Bloch1940,Schleich2001}, the drive field $H_{\rm drive}'(t) = e^{+iH_0 t}H_{\rm drive}(t)e^{+iH_0 t}$ can be written as
\begin{align}
\label{eq:H_drive'}
    H_{\rm drive}' \underset{\rm RWA}{\approx} \frac{\Omega}{2} \ketbra{e}{e}_C \otimes \sigma_M^x,
\end{align}
in the limit where the amplitude $\Omega \ll \omega_M, \kappa,\chi$ is weak, but the detuning is small $\chi \ll \omega_M, \kappa$.

The effective interaction in Eq.~\eqref{eq:H_drive'} can be read as a $\sigma^x$-Hamiltonian applied to the memory qubit conditioned on the control being in the excited state $\ket{e}_C$.
When the control qubit decays from the excited to the ground state, the interaction is autonomously switched off, resulting in an approximation of the gate $U=e^{-i\Omega/(2\Gamma)\sigma^x}$ on the memory qubit.
In Fig.~\ref{fig:RWA_implementation}, we show how the fidelity $\mathcal F(t)=\tr[\ketbra{e}{e}_M \rho_M(t)]$ evolves as a function of time for the model with the RWA, when the memory qubit initially starts in $\ket{g}_M$ and the decay rate is tuned to $\Gamma = \Omega/\pi$ such that an $X$-gate is approximated.
Furthermore, we provide a comparison with the case where the evolution equations are integrated numerically, showing agreement with the RWA.

\paragraph*{Extensions of the model.}
More accurate clocks can be envisaged by replacing the clock qubit by a capped harmonic oscillator such that the clock has to decay through a series of states $\ket{d-1}_C \rightarrow \ket{d-2}_C \dots \ket{1}_C$ before a final transition $\ket{1}_C\rightarrow \ket{0}_C$ turns off the interaction.
Such a series of (effectively classical) decays could for example be engineered using symmetry-selective coupling to thermal environments~\cite{Aamir2022}.
In Fig.~\ref{fig:RWA_implementation}, we show exemplary improvements of the fidelity for $d=4,8,16$.
In principle, genuine quantum clocks beyond purely dissipative processes like e.g.~\cite{Dost2023,Meier2024b} could be employed, though there the practical realizability is still further in the future.
When it comes to the programmability of performing multiple gates, an additional system would be required to encode the punch-card states.
Moreover, different interactions than the one in Eq.~\eqref{eq:H_drive'} would be required, as for example proposed in the work by Mar{\'i}n Guzm{\'a}n et~al.~\cite{Guzman2025}.
The type of coupling between this additional register system and the memory system, however, is of the same type as the conditional coupling between clock and memory system.
Thus, a similar technique like that of the dispersive shift could allow for the program to be implemented on a physical system as well.

\section*{Outlook}
Any truly autonomous device would in addition to the algorithm also need to carry out error correction to counteract the noise that arises from the imperfect timekeeping as well as from other sources not directly modeled within the aQPU.
This, however, is challenged by the fact that common error correction techniques often need syndrome measurements, which are inherently difficult to give a full thermodynamic account for.
There are two possible ways around this for the future: on the one hand one can of course include error correction techniques that require no measurements~\cite{Albert_2019,Kwon_2022,shtanko2023bounds},
but a far more intriguing possibility is to have an inbuilt mesoscopic measurement mechanism that does not fully transition to the classical, yet features the thermodynamically emergent irreversibility of textbook quantum measurements~\cite{Thiago2021}.
With error correction, the logical errors could be made arbitrarily small in principle while the physical errors only have to be below some threshold, which could be achieved even with finite clock accuracy.
While Eq.~\eqref{eq:prec} implies that clock accuracy has to scale with the program's length to achieve a constant fidelity, with error correction, a sufficiently high but constant clock accuracy can guarantee a constant fidelity in the logical subspace.

An additional consideration in this model is the sole contribution of the interaction Hamiltonian to the aQPU's evolution.
In presence of a non-trivial bare Hamiltonian for the memory system, additional considerations regarding energy-conservation would arise such as the use of battery systems~\cite{Chiribella2021,Castellano2025}.
Furthermore, while both the clock as well as the memory register dynamics are in general quantum, the ticking is modeled stochastically making the connection between the control and the memory is incoherent.
The upside of this approach is that the thermodynamic model is self-contained, though on the downside, the timekeeping errors lead to dephasing of the memory register state as already found in~\cite{Aaberg2014,Xuereb2023}, negatively impacting the computational fidelity.
It remains open whether the thermodynamic resources could be used directly for the quantum computation akin to how classical computation is performed in the framework of thermodynamic computing~\cite{LipkaBartosik2024,Rolandi2026,Melanson2025}.

An alternative, albeit experimentally more challenging way of connecting the quantum clock to the computation could use a coherent interaction instead of an stochastic ticking, as for example suggested in~\cite{Jordan2017}.
Another fascinating approach to the autonomous control of quantum computations ~\cite{Woods2024} appeared during revisions of this manuscript. This works uses fully coherent control to improve the computational fidelity while leaving the autonomous reset of the control clock open.

To close, whilst the aQPU is useful as a theoretical framework for analyzing the thermodynamics of quantum computation, it also dares us to think about how necessary external control is for quantum computation.
Just as Feynman challenged the physics community to build the smallest possible classical computer~\cite{feynman1959plenty}, the aQPU challenges the community to build the smallest quantum mechanical computing device.

\acknowledgements
The authors thank Phila Rembold for suggestions on the clarity of the text, Ralph Silva, Nuriya Nurgaleiva, Marek Gluza, Yuri Minoguchi and Gerard Milburn for insightful disussions, and Greg Kuperberg for comments on the Solovay-Kitaev Theorem.
This project is co-funded by the Austrian Science Fund FWF (Grant DOI:~10.55776/COE1), and by the European Union (Quantum Flagship project ASPECTS, Grant Agreement No. 101080167).
Views and opinions expressed are however those of the authors only and do not necessarily reflect those of the European Union, REA or UKRI. Neither the European Union nor UKRI can be held responsible for them.
J.X., P.E. and M.H. acknowledge funding from the European Research Council (Consolidator grant `Cocoquest’ 101043705). 
The authors acknowledge TU Wien Bibliothek for financial support through its Open Access Funding Programme.
Lastly, we acknowledge time itself, not only for being such a physically and philosophically interesting notion, but also for allowing us the space to think about and enjoy the ideas presented in this manuscript.

\providecommand{\noopsort}[1]{}\providecommand{\singleletter}[1]{#1}%

\clearpage

\appendix
\onecolumngrid
\setcounter{page}{1}
\begin{center}
    \large \bfseries Technical Matter
\end{center}
\vspace{1em}
\twocolumngrid
\let\addcontentsline\oldaddcontentsline %

\begingroup
\parskip=0pt
\setcounter{tocdepth}{2}
\tableofcontents
\endgroup

\section{\label{SM:properties_autonomous_clock}Properties of autonomous quantum clocks}
In this Section, we provide background regarding the use of autonomous quantum clocks within the aQPU.
In Sec.~\ref{appendix:applicability_ME}, we justify why the use of the master equation is adequate in the present context and in Sec.~\ref{appendix:ticking_probabilities}, some results from the clock literature are recapitulated for the purpose of this work.

\subsection{\label{appendix:applicability_ME}Applicability of the master equation}
The clockwork's internal evolution is generated by the Lindblad super-operator $\mathcal L_C,$ which can be brought into the form~\cite{Woods2021,Silva2023},
\begin{align}
\label{eq:LC_decomposition}
    \mathcal L_C = -i\left[H_C,\,\cdot\, \right] + \sum_\ell \left(L_\ell \,\cdot\, L_\ell^\dagger - \frac{1}{2}\left\{L_\ell^\dagger L_\ell,\,\cdot\,\right\}\right),
\end{align}
with $H_C$ representing the coherent part of the clockwork evolution and the sum with the operators $L_\ell$ the incoherent part of the evolution.
This second part is generally necessary to drive the clock, e.g., if the clock's ticks dissipate energy, this energy has to be replenished, and the operators $L_\ell$ can in principle do this, contrarily to $H_C$ which (by definition) acts on the system in an energy-preserving way.
The master equation description is merely an effective description of the system (here: the  clock) ignoring the degrees of freedom of a bath or environment.
Microscopically, the system interacts with some environment and only under specific assumptions is it possible to derive an effective description of the system with a master equation, where the detailed environment behavior can be ignored, and where the description used is a faithful approximation of the true dynamics of the system.
A rich variety of literature has been written about this topic, of which some choices are~\cite{Carmichael2002,Walls2008,Hofer2017}.
For the internal clock degrees of freedom, the applicability of the master equation has been subject of several studies and in particular~\cite{Woods2021,Silva2023} have shown from first principles that a master equation description is compatible with a wide range of relevant clock models.

In our work, however, we additionally couple the clock system to the computational degrees of freedom by means of the tick register and the interaction Hamiltonian.
There, it is a priori not clear, whether working within \textit{local} master equation picture~\cite{Hofer2017} as we do in our work is justified.
For the local master equation, the coupling between the bath and the system must be much stronger than the interactions between the constituents of the system.
For the aQPU, it turns out that the regime where high computational fidelity is achieved coincides with the regime where this relation between energy-scales is fulfilled.
For high computational fidelity as discussed in the main text, the master clock needs to have high accuracy $N\gg L \phi_{\rm max}^2$, with $L$ the program length.
On the other hand, the accuracy of a ticking clock is always limited by it's ticking rate $\nu$ as shown in~\cite{Meier2023}.
This bound can be written as $N\leq \|J\|_\infty^2 /\nu^2$, where $J$ is the tick operator  ($J=J_C\otimes \sum_{n=0}^{m-1}\ketbra{n+1}{n}_T$ as defined in the main text),   which together with the above requirement for high fidelity computation implies
\begin{align}
    \frac{\|J \|_\infty^2}{\nu^2} \gg L\phi_{\rm max}^2.
\end{align}
Simplifying this expression using that $\tau \|H_{\rm int}\|_\infty \leq \phi_{\rm max}$ (as defined in the main text), we find
\begin{align}
    \|J\|_\infty  \gg \sqrt{L} \|H_{\rm int}\|_\infty.
\end{align}
Since $L\geq 1$, this indeed shows that the interaction Hamiltonian only weakly couples the clock systems to the remaining aQPU, thus justifying the use of the master equation description.

\subsection{\label{appendix:ticking_probabilities}Ticking probabilities}
To characterize the clock's tick distributions, we can consider the joint evolution of the clockwork $C$ and the tick register $T$, generated by the evolution equation
\begin{align}
\label{eq:dot_rho_CR}
    \dot \rho_{CT} = (\mathcal L_C + \mathcal D_{J})\rho_{CT}.
\end{align}
The term $\mathcal L_C$ as in Eq.~\eqref{eq:LC_decomposition} generates the internal clockwork evolution while $\mathcal D_J$ generates the ticks.
The trace $\tr\left[\left(\mathds 1_C\otimes\ketbra{n}{n}_{T}\right)\rho_{CT}(t)\right]$ can be interpreted as the probability ${\rm Pr}[N(t)=n]$ of the clock having ticked exactly $n$ times at time $t$~\cite{Silva2023}.
Normalization of the quantum state $\rho_{CT}(t)$ ensures that the probability ${\rm Pr}[N(t)=n]$ is normalized with respect to a sum over all non-negative integers $n\geq 0$, i.e., $1=\sum_{n\geq 0}{\rm Pr}[N(t)=n]$.
The  state of the  joint $CT$ system can be decomposed as~\cite{Silva2023},
\begin{align}
\label{eq:rho_CR}
    \rho_{CT}(t) = \sum_{0\leq n \leq m} \rho_C^{(n)}(t)\otimes\ketbra{n}{n}_T,
\end{align}
which allows us to simplify the trace expression as
\begin{align}
    \label{eq:P_Nt_n}
    {\rm Pr}[N(t)=n] = \tr\left[\rho_C^{(n)}(t)\right].
\end{align}
This probability ensemble samples over the possible number of times $n$ that the clock has ticked: here, the number $n$ fluctuates.
If instead, we are asking about the probability that the $n$th tick occurs before time $t,$ we are in a different ensemble, where $n$ is fixed, but the time $t$ fluctuates.
This probability can be denoted by ${\rm Pr}[T_n \leq t].$
We may read the inequality $T_n\leq t$ as the tick time $T_n$ of the $n$th tick lies before $t$, i.e., it is smaller than $t$.
We now present a number of basic properties of the tick probability density modified from~\cite{Silva2023}, which are necessary to understand the main result of this work.
Firstly, Lemma~\ref{lemma:TickNumberConversion} is the following relation:

\begin{restatable}[]{lemma}{TickNumberConversion}
\label{lemma:TickNumberConversion}
Let $N(t)$ be the random variable describing the number of ticks of a clock at time $t$ and let $T_n$ be the random variable describing the time at which the $n$th tick occurs. Then, the following transformation
\begin{align}
\label{eq:P_Nt_n_difference}
    {\rm Pr}[N(t)=n] = {\rm Pr}[T_n\leq t] - {\rm Pr}[T_{n+1}\leq t],
\end{align}
converts between the two ensemble formulations.
\end{restatable}

\begin{proof}
Observe that the events $\{N(t)=n\}$ and $\{T_n\leq t \,\wedge\, T_{n+1}\geq t\}$ coincide, because $n$ ticks at time $t$ is the case if and only if the $n$th tick happened before time $t$ and the $(n+1)$-st tick happens after $t.$
Therefore, we can write
\begin{align}
	{\rm Pr}[N(t)=n] &= {\rm Pr}[T_n\leq t \,\wedge\, T_{n+1}\geq t] \\
	&={\rm Pr}[T_n\leq t] + {\rm Pr}[T_{n+1}\geq t] \nonumber \\
 &\qquad- {\rm Pr}[T_n\leq t \,\vee\, T_{n+1}\geq t],\label{eq:PNT_second_trsf}
\intertext{where the second line~\eqref{eq:PNT_second_trsf} uses the addition rule. Now we can use that the probability ${\rm Pr}[T_n\leq t \,\vee\, T_{n+1}\geq t],$ that the $n$th tick happens before time $t$ or the $(n+1)$-st tick happens after time $t$ is trivially 1. Thus, we can continue the derivation from before}
    \text{\eqref{eq:PNT_second_trsf}}&={\rm Pr}[T_n\leq t] + {\rm Pr}[T_{n+1}\geq t] -1 \\
	&={\rm Pr}[T_n\leq t] - {\rm Pr}[T_{n+1}\leq t],
\end{align}
which proves the Lemma.
\end{proof}

Another useful identity concerning the tick probability density of the $n$th tick ${\rm Pr}[T_n = t] = \partial_t {\rm Pr}[T_n\leq t]$ follows.
Taking the derivative of Eq.~\eqref{eq:P_Nt_n_difference} gives
\begin{align}
\label{eq:tr_dot_rho_n}
    \tr\left[\dot \rho_C^{(n)}(t)\right] =  {\rm Pr}[T_n= t] - {\rm Pr}[T_{n+1}= t].
\end{align}
Using the evolution equations in Eq.~\eqref{eq:dot_rho_CR}, we can express the time derivative in Eq.~\eqref{eq:tr_dot_rho_n} as
\begin{align}
    \dot\rho^{(n)}(t) &= \mathcal L_C\left[\rho_C^{(n)}(t)\right]  - \frac{1}{2}\left\{J^\dagger J,\rho_C^{(n)}\right\} \nonumber \\
    &\qquad+J \rho_C^{(n-1)}(t)J^\dagger,\label{eq:dot_rho_n}
\end{align}
and then, we can directly write the tick probability as a function of the state, as detailed in Lemma~\ref{lemma:JumpTickProbDensity}:

\begin{restatable}[]{lemma}{JumpTickProbDensity}
\label{lemma:JumpTickProbDensity}
    Given the clock model with Ansatz as defined in Eq.~\eqref{eq:rho_CR} and evolution equations~\eqref{eq:dot_rho_CR}, the tick probability density
    \begin{align}
        {\rm Pr}[T_{n+1}=t] = \frac{d}{dt} {\rm Pr}[T_{n+1}\leq t],
    \end{align}
    can be obtained from the state $\rho_{CR}(t)$ as follows,
    \begin{align}
        \tr\left[J^\dagger J \rho_C^{(n)}(t)\right] &= {\rm Pr}[T_{n+1}=t].
    \end{align}
\end{restatable}
\begin{proof}
We show the statement by induction in $n$. The base case: For $n=0$ the statement is a consequence of Lemma~\ref{lemma:TickNumberConversion}. By definition, ${\rm Pr}[T_0\leq t]=1$ for all $t\geq 0$ and thus, the previous Lemma gives
\begin{align}
    \tr\left[\rho_C^{(0)}(t)\right] = 1 - {\rm Pr}[T_1 \leq t].
\end{align}
From this equation, we just have to take the derivative and insert the expression for $\dot \rho_C^{(0)}(t)$ (we do not write out the $t$ argument explicitly),
\begin{align}
    -\frac{d}{dt}{\rm Pr}[T_1\leq t] &= \tr\left[\mathcal L_C\left[\rho_C^{(0)}\right] - \frac{1}{2}\left\{J^\dagger J,\rho_C^{(0)}\right\}\right] \\
    & = -\tr\left[J \rho_C^{(0)} J^\dagger\right],
\end{align}
where we have used cyclicity of the trace in the second line and the fact that $\mathcal L_C[\rho]$ is always trace-less. This proves the base case.

The induction step: we assume that the theorem holds for some value of $n,$ then, we can show (again using Lemma~\ref{lemma:TickNumberConversion}) that it holds for $n+1.$ We look again at ${\rm Pr}[N(t)=n+1]$ which is the trace of $\rho^{(n+1)}(t).$ Taking the time-derivative of that state (see Eq.~\eqref{eq:dot_rho_n}), we find
\begin{align}
    \dot\rho^{(n+1)}(t) &= \mathcal L_C\left[\rho_C^{(n+1)}(t)\right]  - \frac{1}{2}\left\{J^\dagger J ,\rho_C^{(n+1)}\right\} \\
    &\qquad+J\rho_C^{(n)}(t)J^\dagger.
\end{align}
Now, we can trace and on the left-hand-side, we get $\partial_t {\rm Pr}[N(t)=n+1]$ where we can invoke Lemma~\ref{lemma:TickNumberConversion}.
On the right-hand-side, we can use the induction hypothesis and we can replace $\tr\left[J\rho_C^{(n)}J^\dagger\right]$ by ${\rm Pr}[T_{n+1}=t],$ which leaves us with
\begin{align}
    &{\rm Pr}[T_{n+1}=t] - {\rm Pr}[T_{n+2}=t] \\ &\qquad = {\rm Pr}[T_{n+1}=t] - \tr\left[J \rho_C^{(n+1)}J^\dagger \right].\nonumber 
\end{align}
Simplifying this expression yields the statement for $n+1,$ completing the induction step.
By induction, the desired statement follows for all values of $n\geq 0$ which is all we wanted to show.
\end{proof}

\section{\label{sec:results}Technical Results}
In this section, more detailed formulations of the results presented in the main text are discussed, including detailed proofs.
We start in Sec.~\ref{sec:aQPU_universality} by explicitly showing that the aQPU is a universal quantum computer if the master clock is perfect.
Next, Sec.~\ref{sec:error_propagation_nonideal_clocks} explores how errors propagate in case the clock is non-ideal together with a probability theoretic excursion to acquire the necessary techniques to prove the results on error propagation.
Then, in Sec.~\ref{sec:termodynamic_cost_of_computation} the computational notion of fidelity is connected to the thermodynamic cost of computation.
Finally, in Sec.~\ref{sec:speed_vs_fidelity} we discuss how compiling programs on an aQPU with access to a finite set of Hamiltonians and a non-ideal clock presents a trade-off between the speed and fidelity of a computation.

\subsection{\label{sec:aQPU_universality}The aQPU is universal for ideal clocks}
For imperfect clocks, it is of course not expected that that the aQPU computes $V_{\mathcal A}$ perfectly on the memory system.
As we know from previous works~\cite{Aaberg2014,Ball2016,Xuereb2023} non-ideal timing of unitary gates using a classical tick register leads to dephasing.

As a first step, we show that if the master clock is ideal, the aQPU generates $V_{\mathcal A}$ exactly.
Ideal in this setting means that the distribution of ticks of the master clock is perfectly regular, i.e., if we set the average time between two ticks as $\tau,$ then,
\begin{align}
    \label{eq:P_tick_ideal}
    {\rm Pr}[T_n=t] = \delta(n\tau -t).
\end{align}
We use the terms \textit{perfect} and \textit{ideal} clock interchangeably.
Similarly, for the number of ticks $N(t),$ we have a window-function-like probability given by the expression
\begin{align}
    \label{eq:PN_ideal}
    {\rm Pr}[N(t)=n] = \Theta(t\geq n\tau) - \Theta(t\leq (n+1)\tau),
\end{align}
with $\Theta$ being the Heaviside step function.
Formally, we can state the result as follows:

\begin{restatable}[Universality for perfect clocks]{lemma}{universalitytheorem}
\label{lemma:universalitytheorem}
    Let the aQPU model be defined by the Lindbladian $\mathcal L_\mathrm{aQPU}$ as in Eq.~\eqref{eq:L_aQPU} with access to a finite number of Hamiltonians that generate a universal gate set $\mathcal V$ and let $\mathcal A$ by any finite program defined on $\mathcal V$. If the master clock with initial state $\rho_{C}^\mathrm{init}$ is ideal with tick time $\tau$ we have that
    \begin{align}
    \label{eq:universality_eq}
        &\tr_{CTI}\Big[e^{t \mathcal L_{\mathrm{aQPU}}} \left(\rho_{C}^\mathrm{init}\otimes\ketbra{0}{0}_{T}\otimes\ketbra{\mathcal A}{\mathcal A}_I\otimes \rho_M^\mathrm{init}\right) \Big] \nonumber \\
        &\quad = V_{\mathcal A} \rho_M^\mathrm{init} V_{\mathcal A}^\dagger,
    \end{align}
    for $t\geq L\tau$ large enough.
\end{restatable}
This first result shows that our model in the ideal limit can recover the universality that previous models mentioned in the introduction~\cite{Benioff1980,Benioff1982,deutsch,vazirnai,Feynman2023,KSV02} achieved.
With our open system's model however, we can now go beyond this mere existence result and explore how well one can reach universality in realistic scenarios with limited resources, which has so far not been explored in the setting of autonomous quantum computation.
Before we start with the proof of Lemma~\ref{lemma:universalitytheorem}, we focus on the following two preliminaries, at the full level of generality of our model, i.e., we will assume a general master clock that may very well by non-ideal.
\begin{itemize}
    \item The state-structure of $\rho(t)$:  we show that the state can be expanded as an incoherent mixture over the states $\ketbra{n}{n}_{T}$ of the clock's tick register. We formalize this in Lemma~\ref{lemma:statestructure}.
    \item Given this specific structure, we can first solve for the clock dynamics and then secondly solve the memory system dynamics separately. See Prop.~\ref{prop:targetrecursion}.
\end{itemize}
The aQPU initially starts in a uncorrelated state of clock, tick register, instruction register and memory system.
Due to the structure of the time-evolution generator $\mathcal L_\mathrm{aQPU},$ the correlations between different tick numbers that build up over time are only classical, such that we find the following simple structure:

\begin{restatable}[State-structure]{lemma}{statestructure}
\label{lemma:statestructure}
Let the initial state defined on the full aQPU Hilbert space $\mathcal H$ be given by
\begin{align}
    \label{eq:rho_aQPU_init}
    \rho^\mathrm{init} = \rho_C^\mathrm{init}\otimes\ketbra{0}{0}_{T}\otimes\ketbra{\mathcal A}{\mathcal A}_I\otimes \rho_M^\mathrm{init},
\end{align}
where $\rho_C^\mathrm{init}$ is an arbitrary initial state on the clockwork $\mathcal H_C$ and $\rho_C^\mathrm{init}$ an arbitrary initial state on the memory system.
Then, at any point in time $t$ the state $\rho(t) = e^{\mathcal L_\mathrm{aQPU} t}\rho^\mathrm{init}$ is given by
\begin{align}
    \label{eq:rho_CRT}
    \rho(t) = \sum_{0\leq n \leq m} \rho_C^{(n)}(t) \otimes \ketbra{n}{n}_{T}\otimes \ketbra{\mathcal A}{\mathcal A}_I\otimes \rho_M^{(n)}(t).
\end{align}
\end{restatable}

\begin{proof}
The three terms that generate the evolution as in~\eqref{eq:L_aQPU} are given by a clockwork term $\mathcal L_{C},$ a tick-generating term $J$ and a three-body interaction term $H_\mathrm{int}.$
By definition, these terms do not create any coherence between different tick numbers, and thus it is sufficient to show that for any state $\rho$ of the form
\begin{align}
\label{eq:rho_form}
    \rho = \rho_C\otimes \ketbra{n}{n}_{T} \otimes \ketbra{\mathcal A}{\mathcal A}_I \otimes \rho_M,
\end{align}
we have $\mathcal L \rho$ is a sum of terms like the one above but possibly different $\rho_C,\rho_M$ and $n.$
The reason this suffices is that the time-evolution is generated by $\mathcal L$ as in Eq.~\eqref{eq:L_aQPU}.
If we thus start with a state like that in~\eqref{eq:rho_form}, at any future point in time $t$, the state $\rho(t) = e^{\mathcal L_\mathrm{aQPU}t}\rho$ will be a sum of terms of said form.
But this is exactly the statement of the Lemma to prove.
Thus, let us look term by term at $\mathcal L_\mathrm{aQPU}$:
\begin{align}
    \mathcal L_C[\rho] = \mathcal L_C[\rho_C] \otimes \ketbra{n}{n}_{T} &\otimes \ketbra{\mathcal A}{\mathcal A}_I \otimes \rho_M,
\end{align}
which is of the desired form.
The tick generating term becomes,
\begin{widetext}
\begin{align}
    \mathcal D_J[\rho] = \left(- \frac{1}{2}\left\{J_C^\dagger J_C,\rho_C\right\}\otimes \ketbra{n}{n}_{T} + J_C\rho_C J_C^\dagger \otimes\ketbra{n+1}{n+1}_{T}\right) \otimes \ketbra{\mathcal A}{\mathcal A}_I \otimes \rho_M,
\end{align}
\end{widetext}
which also aligns with the required form~\eqref{eq:rho_form}. Finally, the interaction terms yields
\begin{align}
    -i[H_{\rm int},\rho] = -i\rho_C \otimes\ketbra{n}{n}_{T} \otimes \ketbra{\mathcal A}{\mathcal A}_I \otimes \left[H_{M,a_n},\rho_M\right],
\end{align}
where $a_n$ is the $n$th entry in the program $\mathcal A$.
This is also of the desired form and together with the initial remark proves the Lemma.
\end{proof}

Now that we know the state-structure of the aQPU at all times $t$, we can insert it as an Ansatz into the evolution equations generated by $\mathcal L_{\rm aQPU}$ and see how the reduced state $\rho_M(t)$ of the memory system evolves.
Without loss of generality we may assume that for all $n$ and for all $t$, the memory's state is normalized
\begin{align}
\label{eq:tr_rho_T_1}
    \tr\left[\rho_M^{(n)}(t)\right] = 1.
\end{align}
As a consequence we can keep using the identity from Eq.~\eqref{eq:P_Nt_n} for the probability ${\rm Pr}[N(t)=n]$, giving us an explicit way to determine the memory's state,
\begin{align}
\label{eq:rho_T}
    \rho_M^{(n)}(t) = \frac{\tr_{CTI}\Big[\big(\mathds 1_{CIM}\otimes \ketbra{n}{n}_{T}\big) \rho(t) \Big]}{{\rm Pr}[N(t)=n]}.
\end{align}
The denominator is to ensure normalization from Eq.~\eqref{eq:tr_rho_T_1} by countering the trace over the clock state as in Eq.~\eqref{eq:P_Nt_n}.
The missing piece towards showing Lemma~\ref{lemma:universalitytheorem} is the answer to the question: how does the memory state evolve? The following proposition provides the answer.

\begin{restatable}[Memory recursion relation]{proposition}{targetrecursion}
\label{prop:targetrecursion}
The memory system's state $\rho_M^{(n)}(t)$ at parameter time $t$, conditioned on $n$ ticks having occurred takes the following form,
\begin{align}\label{eq:rho_t_ansatz}
    \rho_M^{(n)}(t) = \int_0^t ds\, \xi(t,s) V_{a_n}(t-s) \rho_M^{(n-1)}(s) V_{a_n}(t-s)^\dagger.
\end{align}
The function $\xi(t,s)$ describes the probability distribution of the $n$th tick occurring at time $s$ normalized on the interval $s\in[0,t]$,
\begin{align}\label{eq:xi_ts_and_pt}
    \xi(t,s) =p(s) \exp\left({-\int_s^t d\tau p(\tau)}\right),
\end{align}
with $p(\tau)={\rm Pr}[T_n=t]/{\rm Pr}[N(t)=n]$.
Moreover, the unitary $V_{a_n}(t)$ is the propagator at time $t$ generated by the Hamiltonian $H_{M,a_n},$ i.e., $V_{a_n}(t)=\exp\left(-iH_{M,a_n}t\right).$
\end{restatable}

Before proceeding with the proof of Prop.~\ref{prop:targetrecursion}, we note that Eq.~\eqref{eq:rho_t_ansatz} is a generalization of the expression from~\cite{Xuereb2023} for the impact of imperfect time-keeping on the evolution of a quantum system under a controlled unitary.
Here, the evolution is averaged over distribution $\xi(t,s)$,
\begin{align}
    V_{a_n}(t-s) \rho_M^{(n-1)}(s) V_{a_n}(t-s)^\dagger,
\end{align}
which is the state at $\rho_M^{(n-1)}(s)$ after $n-1$ ticks at time $s$, when the $n$th tick occurs exactly at time $s$ and evolves for another time $t-s$ according to the propagator generated by $H_{M,a_{n}}$.
For consistency, we verify that $\xi(t,s)$ is indeed normalized: Write $\xi(t,s) = \partial_s \zeta(t,s),$ where
\begin{align}
    \zeta(t,s) &= \exp\left({-\int_s^t d\tau p(\tau)}\right).
\end{align}
This allows analytically calculating the integral of $\xi(t,s)$ and therefore also the normalization condition because
\begin{align}
    \int_0^t ds\,\xi(t,s)&=\int_0^t ds\,\left(\partial_s\zeta(t,s)\right) \\
    &= \zeta(t,t)=1,
\end{align}
which was our claim, that $\xi(t,s)$ is a genuine probability distribution over $t$.
We proceed with proving Prop.~\ref{prop:targetrecursion}.

\begin{proof}
The proof of this statement consists of the two steps already pointed out in the main text: first, we show that $\rho_M^{(n)}(t)$ is governed by the evolution equations~\eqref{eq:dot_rho_T}, which, for completeness, we recall here,
\begin{align}
    \dot \rho_M^{(n)}(t) = -i\big[H_{M,a_n},\rho_M^{(n)}&(t)\big] +\nonumber\\ &p(t)\big(\rho_M^{(n-1)}(t)-\rho^{(n)}_M(t)\big),       \label{eq:dot_rho_T_appendix}
\end{align}
with $p(t)$ defined as
\begin{align}
\label{eq:pt_appendix}
    p(t) := \frac{{\rm Pr}[T_n=t]}{{\rm Pr}[N(t)=n]}.
\end{align}
The second step of the proof is to verify that the recursion relation~\eqref{eq:rho_t_ansatz} indeed solves the evolution equation~\eqref{eq:dot_rho_T_appendix}.
Without further ado, we get started with the first step.

Recall Eq.~\eqref{eq:rho_T} multiplied by ${\rm Pr}[N(t)=n]$ and take the time-derivative on both sides.
On the right-hand side, we get $\dot \rho(t)$ which can be replaced by $\mathcal L_\mathrm{aQPU}[\rho(t)].$
The three resulting terms are
\begin{itemize}
    \item Clockwork term,
    \begin{align}
    \label{eq:cw_term}
        \tr_{CTI}\Big[\big(\mathds 1_{CIM}\otimes\ketbra{n}{n}_{T}\big)\mathcal L_C[\rho(t)]\Big].
    \end{align}
    \item Ticking term,
    \begin{align}
    \label{eq:tick_term}
        \tr_{CTI}\Big[\big(\mathds 1_{CIM}\otimes\ketbra{n}{n}_{T}\big)\mathcal D_J[\rho(t)]\Big].
    \end{align}
    \item Interaction term,
    \begin{align}
    \label{eq:int_term}
        -i\tr_{CTI}\Big[\big(\mathds 1_{CIM}\otimes\ketbra{n}{n}_{T}\big)[H_{\rm int},\rho(t)]\Big].
    \end{align}
\end{itemize}
The clockwork term~\eqref{eq:cw_term} is trivially zero, because $\mathcal L_C$ is traceless.
The tick-generating term on the other hand yields non-trivial contributions.
Let us thus first calculate $\mathcal D_J[\rho(t)]$ in full:
\begin{widetext}
\begin{align}
\label{eq:L_tick_rho_full}
    \mathcal D_J[\rho(t)] &= \sum_{n\geq 0}\mathcal D_J\left[\rho_C^{(n)}(t) \otimes\ketbra{n}{n}_{T}\otimes \ketbra{\mathcal A}{\mathcal A}_I\otimes \rho_M^{(n)}(t) \right] \\
    &= \sum_{n\geq 0}\left(J_C \rho_C^{(n)} J_C^\dagger \otimes\ketbra{n+1}{n+1}_{T} - \frac{1}{2}\left\{J_C^\dagger J_C^{(n)},\rho_C^{(n)}\right\}\otimes\ketbra{n}{n}_{T}\right) \otimes \ketbra{\mathcal A}{\mathcal A}_I \otimes\rho_M^{(n)}(t)
\end{align}
\end{widetext}
by taking the trace as in Eq.~\eqref{eq:tick_term}, we find
\begin{align}
    \text{\eqref{eq:tick_term}} = {\rm Pr}[T_n = t]\rho_M^{(n-1)}(t) - {\rm Pr}[T_{n+1}=t]\rho_M^{(n)}(t).
\end{align}
Next, we look at the contribution from the interaction as written out in Eq.~\eqref{eq:int_term}. 
To understand this term better, recall $H_\mathrm{int}$ is a sum over $0\leq n \leq M$ of terms of the form
\begin{align}
    \mathds 1_C\otimes\ketbra{n}{n}_{T}\otimes \mathds 1_{I_{m(\neq n)}}\otimes \sum_{k=1}^K \ketbra{k}{k}_{I_n} \otimes H_{M,k}.
\end{align}
Lemma~\ref{lemma:statestructure} ensures that $\rho(t)$ is diagonal with respect to the tick register states $\ket{n}_{T}$.
Thus, the projector $\mathds 1_{CIM}\otimes\ketbra{n}{n}_{T}$ in Eq.~\eqref{eq:int_term} picks out the $n$th term in the sum of $\rho(t)$ (in the notation of~\eqref{eq:rho_t_ansatz}) and similarly,  the punch card state $\ketbra{\mathcal A}{\mathcal A}_I$ of $\rho(t)$ picks out the interaction term where $k=a_n$ and this leads to the following expression:
\begin{widetext}
    \begin{align}
        \text{\eqref{eq:tick_term}} &= -i\tr_{CTI}\left[\rho_C^{(n)}(t)\otimes \ketbra{n}{n}_{T}\otimes \left[\mathds 1_{I_{m(\neq n)}}\otimes\sum_{k=1}^K \ketbra{k}{k}_{I_n}\otimes H_{M,k},\ketbra{\mathcal A}{\mathcal A}_I\otimes\rho_M^{(n)}(t)\right]\right] \\
        &= -i\tr_{CTI}\left[\rho_C^{(n)}(t)\otimes \ketbra{n}{n}_{T}\otimes\ketbra{\mathcal A}{\mathcal A}_I\otimes \left[H_{M,a_n},\rho_M^{(n)}(t)\right]\right] \\
        &= -i {\rm Pr}[N(t)=n]\left[H_{M,a_n},\rho_M^{(n)}(t)\right].
    \end{align}
Finally, we can add all the terms~\eqref{eq:cw_term}, \eqref{eq:tick_term} and~\eqref{eq:int_term} together.
For the time-derivative of ${\rm Pr}[N(t)=n]\rho_M^{(n)}(t)$, we can explicitly calculate by using Lemma~\ref{lemma:TickNumberConversion},
    \begin{align}
        &({\rm Pr}[T_n=t] - {\rm Pr}[T_{n+1}=t])\rho_M^{(n)}(t) + {\rm Pr}[N(t)=n]\dot \rho_M^{(n)}(t) \\
        & \qquad = {\rm Pr}[T_n = t]\rho_M^{(n-1)}(t) - {\rm Pr}[T_{n+1}=t]\rho_M^{(n)}(t) -i {\rm Pr}[N(t)=n]\left[H_{M,a_n},\rho_M^{(n)}(t)\right],
    \end{align}
which we can simplify to
\begin{align}
\label{eq:dot_rho_T_non_singular_rom}
     {\rm Pr}[N(t)=n]\dot \rho_M^{(n)}(t) =-i {\rm Pr}[N(t)=n]\left[H_{M,a_n},\rho_M^{(n)}(t)\right] + {\rm Pr}[T_n = t]\left(\rho_M^{(n-1)}(t) - \rho_M^{(n)}(t)\right).
\end{align}
\end{widetext}
Dividing both sides by ${\rm Pr}[N(t)=n],$ we find the desired evolution equation as claimed in Eq.~\eqref{eq:dot_rho_T_appendix}.
This completes the first step of the proof.

As for the second step, we want to verify that the expression in~\eqref{eq:rho_t_ansatz} solves said evolution equations.
We insert~\eqref{eq:rho_t_ansatz} into the evolution equations~\eqref{eq:dot_rho_T_appendix} for this. To simplify the proof, we abbreviate the notation in the following way
\begin{align}
    v(t) &\equiv \rho_M^{(n)}(t), \label{eq:def_vt}\\
    w(t) &\equiv \rho_M^{(n-1)}(t),\label{eq:def_wt}\\
    A &\equiv -i\left[H_{M,a_n},\circ\,\right], \label{eq:def_A}
\end{align}
and $p(t)$ as in Eq.~\eqref{eq:xi_ts_and_pt} from the proposition. In this notation, the evolution equations in Eq.~\eqref{eq:dot_rho_T} read
\begin{align}
    \dot v(t) = Av(t) + p(t)(w(t)-v(t)),
\end{align}
and the ansatz from Eq.~\eqref{eq:rho_t_ansatz} can be recast into
\begin{align}\label{eq:vs_ansatz}
    v(t) = \int_0^t \underbrace{p(s)\exp\left(-\int_s^t d\tau p(\tau)\right)}_{=\xi(t,s)} e^{A(t-s)}w(s).
\end{align}
All we need to do now, is to take the time derivative of $v(t)$ as defined in Eq.~\eqref{eq:vs_ansatz}. The product rule will give us three contributions,
\begin{align}\label{eq:dot_vt}
    \dot v(t) &= \xi(t,t)w(t) + \int_0^t ds\,\left(\partial_t \xi(t,s)\right) e^{A(t-s)}w(s) + Av(t).
\end{align}
The partial derivative of $\xi(t,s)$ with respect to $t$ can be calculated by using the definitions from Eq.~\eqref{eq:xi_ts_and_pt} to give
\begin{align}\label{eq:partial_t_xi}
    \partial_t\xi(t,s)&= -\xi(t,s)\partial_t \int_s^t d\tau p(\tau) \\
    &= -\xi(t,s) p(t).
\end{align}
Essentially, this result allows us to re-express the middle term on the right-hand side of Eq.~\eqref{eq:dot_vt},
\begin{align}\label{eq:int_partial_xi}
     \int_0^t ds\,\left(\partial_t \xi(t,s)\right) e^{A(t-s)}w(s) = -p(t) v(t).
\end{align}
Together with the identity $\xi(t,t)=p(t),$ we can use Eq.~\eqref{eq:dot_vt}, insert Eq.~\eqref{eq:int_partial_xi} into the middle term and we finally recover
\begin{align}
    \dot v(t) = p(t)w(t) - p(t)v(t) + Av(t),
\end{align}
which is exactly the expression from Eq.~\eqref{eq:dot_vt}. This proves that the ansatz as defined in Eq.~\eqref{eq:vs_ansatz} solves this differential equation; moreover, if we revert our notation change from Eqs.~\eqref{eq:def_vt},\eqref{eq:def_wt} and~\eqref{eq:def_A}, we recover the expression from the proposition, which is all we wanted to show.
\end{proof}

A special case of Prop.~\ref{prop:targetrecursion} is the case $n=0,$ which describes the evolution of the memory system conditioned on no ticks having occurred yet.
This is the base-case of the recursion relation and there, the evolution equations~\eqref{eq:dot_rho_T} reduce to a standard Schrödinger equation with Hamiltonian $H_{M,a_0}$.
Thus, we have
\begin{align}
    \rho_M^{(0)}(t) = \exp\left(-iH_{M,a_0}t\right)\rho_M^\mathrm{init}\exp\left(+iH_{M,a_0}t\right),
\end{align}
from which we can now derive $\rho_M^{(n)}(t)$ for all $n\geq 1$ by using Prop.~\ref{prop:targetrecursion}.
An interesting feature recognizable already at this stage is that while $\rho_M^{(0)}(t)$ evolves unitarily, $\rho_M^{(1)}(t)$ evolves according to a mixed unitary channel due to the uncertainty of when the master clock produces its first tick. This trend continues as more operations are concatenated, but the details of this analysis will come in Sec.~\ref{sec:error_propagation_nonideal_clocks}.
For the moment, having assembled the requisite tools we will focus on the ideal case where the master clock is perfect and prove Lemma~\ref{lemma:universalitytheorem} stated at the outset of this section.

\begin{proof}[Proof of Lemma~\ref{lemma:universalitytheorem}.]
In the limit where the master clock is ideal, both probabilities ${\rm Pr}[T_n=t]$ and ${\rm Pr}[N(t)=n]$ become distribution-like functions (see Eqs.~\eqref{eq:P_tick_ideal} and~\eqref{eq:PN_ideal}) and we have to take special care when applying Prop.~\ref{prop:targetrecursion} because $p(t)$ is ill-defined.

We thus start with the non-singular expression in Eq.~\eqref{eq:dot_rho_T_non_singular_rom}.
For a perfect clock, this equation reduces to a well-defined differential equation only for values $t\in[n\tau,(n+1)\tau]$ where ${\rm Pr}[N(t)=n]>0$.
There, we find
\begin{align}
     \dot \rho_M^{(n)}(t) &=-i \left[H_{M,a_n},\rho_M^{(n)}(t)\right] \\
     &\quad+ \delta(t-n\tau)\left(\rho_M^{(n-1)}(t) - \rho_M^{(n)}(t)\right).\label{eq:dot_rho_t_singular}
\end{align}
Integrating the expression yields the initial condition $\rho_M^{(n)}(n\tau) = \rho_M^{(n-1)}(n\tau)$ and once we have the initial condition, we see that the singular expression in line~\eqref{eq:dot_rho_t_singular} vanishes, and $\rho_M^{(n)}(t)$ for values $t>n\tau$ follows Schrödinger evolution with Hamiltonian $H_{M,a_n}$,
\begin{align}
        \rho_M^{(n)}(t) = V_{a_n}(t-n\tau)\rho_M^{(n-1)}(n\tau)V_{a_n}(t-n\tau)^\dagger.
\end{align}
We can evaluate this expression for $t=(n+1)\tau$ to find the equation needed for the next term in the recursion,
\begin{align}
    \rho_M^{(n)}((n+1)\tau) = V_{M,a_n}\rho_M^{(n-1)}(n\tau) V_{M,a_n}^\dagger,
\end{align}
where $V_{M,a_n}$ is defined as in Eq.~\eqref{eq:V_Mk_def} of the main text.
Looking at the evolution for $t\geq M\tau,$ where $M$ is the maximum number of steps in the program $\mathcal A,$ we find that $\rho_M^{(M)}(t)$ is given by the concatenation of all the unitaries $V_{M,a_0}, V_{M,a_1},\dots, V_{M,a_{M-1}}$ applied to the initial state $\rho_M^\mathrm{init}$.
In mathematical terms, we get
\begin{align}
    \rho_M^{(M)}(t) = V_{\mathcal A} \rho_M^\mathrm{init} V_{\mathcal A}^\dagger,
\end{align}
for $t\geq M\tau$.
Since for the idealized master clock, the trace $\tr_{CTI}[\rho(t)]$ will simply yield $\rho_M^{(M)}(t)$ if $t\geq M\tau$ this proves Eq.~\eqref{eq:universality_eq}.
We conclude that an aQPU with access to a finite set of Hamiltonians that generate a universal gate set and an ideal clock can generate any unitary on the memory system, which is all we wanted to prove.
\end{proof}

\subsection{\label{sec:error_propagation_nonideal_clocks}Error propagation for non-ideal clocks}
In this technical appendix, we detail what happens if the clock's performance deviates from the ideal one.
For imperfect clocks, the tick times may be some randomly distributed times $t_1,t_2,\dots$ that are close to desired tick times $\tau,2\tau,\dots$ with high probability, but generally not equal (see Fig.~\ref{fig:trajectory} for an illustration).
In the following, we develop the formal tools necessary for quantifying how a non-ideal distribution of the master clock's ticks affects the fidelity of the computation.

\begin{figure}
    \centering
    \includegraphics[width=\columnwidth]{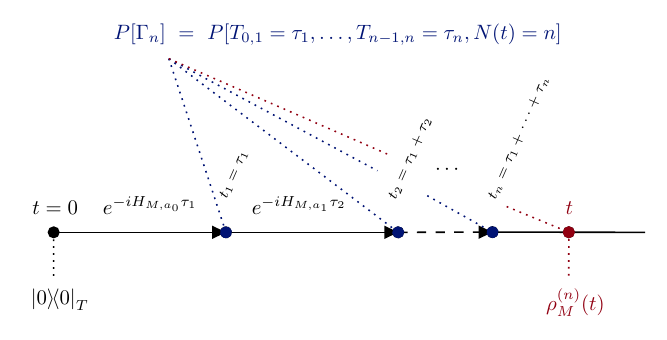}
    \caption{The evolution of the memory system in the aQPU is an average over all possible tick times $T_1,\dots,T_n$ of the master clock.
    In the above figure we illustrate  one possible instantiation of  the evolution of the memory system according to the first Hamiltonian $H_{M,a_0}$ for some time $t_1$ given by the first tick of the master clock.
    The second Hamiltonian $H_{M,a_1}$is applied for some time $t_2-t_1,$ where $t_2$ is the time at which the master clock ticks the second time.
    This scheme continues until the time $t$ at which the memory system's state $\rho_M^{(n)}(t)$ is considered.}
    \label{fig:trajectory}
\end{figure}

\paragraph*{Step 1 -- Stochastic unraveling.}
We can always unravel the evolution of the the aQPU into stochastic trajectories. If we look at evolution time $t$, such a trajectory can have different numbers of total clock ticks $n$. For given $n$, a trajectory $\Gamma_n$ is given by the all the time intervals at which the clock ticked,
\begin{align}
    \label{eq:Gamma_Trajectory}
    \Gamma_n = \left(\tau_1,\tau_2,\dots,\tau_n\right),
\end{align}
where $\tau_k$ is the time between the $(k-1)$st and $k$th tick of the master clock.
The relationship to the tick times $t_1\leq t_2\leq \dots\leq t_{n-1} < t$, is given by $t_k = t_{k-1}+\tau_k,$ where we set $t_0=0$ by definition, and $\tau_n = t - t_{n-1}$.
Let us furthermore consider the program $\mathcal A$ of length at most $M$.
For the following analysis we want to look at an explicit trajectory $\Gamma_n$ of the evolution, where we will find that the computational memory system evolves unitarily,
\begin{align}
    \rho_M(t|\Gamma_n) &= V_{a_{n}}(\tau_n)\cdots V_{a_0}(\tau_1)\rho_M^\mathrm{init} V_{a_0}(\tau_1)^\dagger \cdots V_{a_{n}}(\tau_n)^\dagger.
\end{align}
The probability $p[\Gamma_n]$ that such a trajectory is realized is given by the joint probability that $n$ ticks have occurred at time $t,$ together with the first tick having happened at time $t_1,$ the second at time $t_2,$ etc.\ until the $n$th tick that must have happened at time $t_n.$
Formally, the probability can be expressed as
\begin{align}
\label{eq:p_Gamma_n_prob}
    p[\Gamma_n] &= {\rm Pr}[T_1=t_1,\dots,T_n=t_n,N(t)=n] \\
    &= {\rm Pr}[T_1 = t_1,\dots, T_n=t_n \leq t \leq T_{n+1}].
\end{align}
Summing and integrating over all possible trajectories $\Gamma_n$ then yields the state of the computational memory system at time $t$ on average over all possible times at which the clock could have ticked.
This results in the following expression for the memory system,
\begin{align}
    &\rho_M(t) = \sum_{n=0}^m\int dt_1\cdots dt_n p[\Gamma_n] \rho_M(t|\Gamma_n) \\
    &= \sum_{n=0}^m {\rm Pr}[N(t)=n] \underbrace{\int dt_1\cdots dt_n p[\Gamma_n | N(t) = n] \rho_M(t|\Gamma_n)}_{=\rho_M^{(n)}(t)}, \label{eq:rho_T_resolution_N_integral}
\end{align}
where we have resolved the density matrix as a sum over the possible numbers of ticks in the second line~\eqref{eq:rho_T_resolution_N_integral}.
The stochastic trajectory for the memory system can be derived using more general results on the master equation unraveling of the full aQPU evolution~\cite{Daley2014,Lidar2019}, and then projecting on the memory subsystem.

\paragraph*{Step 2 -- Exponentially concentrated ticks.}
To achieve high fidelity, we are interested in the case of high clock accuracy $N$.
This assumption alone, however, still allows for pathalogical behavior of the clock, since high accuracy does not necessarily impose any constraints on the higher moments of the tick time distribution.
In the following, we impose an additional assumption to bound the tails of the probability distribution by an exponential envelope  (as will be discussed in more detail later on, when elaborating eq. \eqref{eq:cpd}) .
This envelope ensures that the dominant contribution to the probability density comes from times closely centered around $\tau$, and that higher moments of the tick probability density are bounded.
In the following, we introduce and adapt definitions and results from~\cite{Boucheron2013,Rigollet2015} on concentration inequalities that will prove useful later.
The idea behind this is that we want to figure out how the aQPU evolves the memory system's state in case the clock is highly accurate.

We start by considering a generic real random variable $X$ which has without loss of generality mean $\langle X\rangle = 0.$
We say that $X$ is \textit{exponentially concentrated}, if it satisfies
\begin{align}
\label{eq:X_exp_concentrated_def}
    {\rm Pr}[|X|\geq x] = \int_{|x'|>x} dx' {\rm Pr}[X=x'] \leq \alpha e^{-cx},
\end{align}
for two constants $\alpha,c>0$.
Based on this exponential decay condition on the tail of the distribution of $X,$ we can also bound the moments of $X$ and the moment-generating function $M(k)=\langle e^{kX}\rangle,$ which will turn out to be useful later.

\begin{restatable}[Bounded moments]{lemma}{boundedmoments}
    \label{lemma:boundedmoments}
    Let $X$ be exponentially concentrated as defined in~\eqref{eq:X_exp_concentrated_def}.
    Then, the absolute moments of $X$ are bounded as follows,
    \begin{align}
        \left\langle |X|^n\right\rangle \leq \frac{\alpha n!}{c^n}.
    \end{align}
\end{restatable}
\begin{proof}
    Here, we use a modified method following Lemma~5.5 from~\cite{Vershynin2012}. The trick is to define a positive random variable $Z=|X|^n$ and using partial integration (with special care for the boundaries), we can show 
    \begin{align}
        \left\langle |X|^n\right\rangle &= \int_0^\infty dz z {\rm Pr}[Z=z] \\
        &= \int_0^\infty dz {\rm Pr}[Z\geq z].\label{eq:Xabsk_temp}
    \end{align}
    The inequality $Z\geq z$ is equivalent to $|X|\geq x,$ under the change of variables $x^n = z.$ Substitution with $dz = n x^{n-1} dx$ allows us to further reexpress the absolute moment according to
    \begin{align}
        \text{\eqref{eq:Xabsk_temp}} &= \int_0^\infty dx\,n x^n {\rm Pr}[|X|\geq x] \\
        &\leq \int_0^\infty dx\, \alpha n x^n e^{-cx} \\
        &=\frac{\alpha n!}{c^n},
    \end{align}
    by using the definition of the $\Gamma$ function and $\Gamma(n)=(n-1)!$ which is all we wanted to prove for this Lemma.
\end{proof}
Another statement we can make by using the assumption that $X$ is exponentially concentrated is about the moment generating function (MGF); by using Lemma~\ref{lemma:boundedmoments} we can expand this result to the following.
\begin{restatable}[Bounded MGF]{lemma}{boundedMGF}
    \label{lemma:boundedMGF}
    Let $X$ be again exponentially concentrated as in~\eqref{eq:X_exp_concentrated_def}. The resulting MGF is bounded by
    \begin{align}
        M(k) \equiv \left\langle e^{kX}\right\rangle \leq \exp\left(\frac{2\alpha k^2}{c^2}\right),
    \end{align}
    for values $|k|\leq \frac{c}{2}.$
\end{restatable}
\begin{proof}
    We can directly expand the MGF in terms of the moments and employ Lemma~\ref{lemma:boundedmoments}, though note that we also use the fact that the first moment vanishes as we have assumed zero mean for $X$:
    \begin{align}
        M(k) &= \sum_{n\geq 0}\frac{\left\langle (kX)^n \right\rangle}{n!} \\
        &=1 + \sum_{n\geq 2}\frac{\left\langle (kX)^n \right\rangle}{n!} \\
        &\leq 1 + \sum_{n\geq 2}\frac{\left\langle |kX|^n \right\rangle}{n!} \\
        &\leq 1 + \sum_{n\geq 2}\alpha\frac{|k|^n}{c^n} \\
        &=1 + \alpha\frac{k^2}{c^2} \sum_{n\geq 0}\frac{|k|^n}{c^n} \\
        &\leq 1+2\alpha\frac{k^2}{c^2} \\
        &\leq \exp\left(2\alpha \frac{k^2}{c^2}\right)
    \end{align}
    so long as $\left|k\right|\leq \frac{c}{2},$ finalizing our result as desired.
\end{proof}
What we have done so far is analyzed the behavior of the random variable $X.$ In relation to the clock probability distribution, these would be statements about the probability distribution of a single tick.
When considering many ticks, and under the assumptions those ticks are independent and identically distributed (i.i.d.), we would hope that some of the properties about how well the single tick is concentrated would carry over to the sum over many such ticks.
As it turns out, a special case of Bernstein's inequality~\cite{Bernstein1924} provides exactly the desired statement.

\begin{restatable}[Bernstein's inequality -- special case]{lemma}{bernstein}\label{lemma:bernstein}
    Let $X_1,\dots,X_n$ be $n$ i.i.d.\ copies of the exponentially concentrated random variable $X$. Define the sum
    \begin{align}
        \overline X = \sum_{k=1}^n X_k,
    \end{align}
    then for any value of $t,$ we always have
    \begin{align}
        {\rm Pr}\left[\left|\overline{X}\right|\geq x\right] \leq \exp\left(\frac{\alpha n}{2} - \frac{cx}{2}\right)
    \end{align}
\end{restatable}
\begin{proof}
    The Chernoff bound~\cite{Chernoff1952} can be used directly to upper bound the concentration probability for $\overline{X},$
    \begin{align}
    \label{eq:chernoff_P_overline_X}
        {\rm Pr}\left[\left|\overline{X}\right|\geq x\right] \leq M_n(k)e^{-kx},
    \end{align}
    where $M_n(k)$ is the MGF of $\overline{X}.$ By using the i.i.d.\ property of the random variables $X_1,\dots,X_n$ that sum up to $\overline{X},$ we can can bound $M_n(k)$ from above by using Lemma~\ref{lemma:boundedMGF} and the fact that all random variables $X_i$ are exponentially concentrated according to~\eqref{eq:X_exp_concentrated_def}. We find,
    \begin{align}
        M_n(k) = M(k)^n \leq \exp\left(\frac{2\alpha n k^2}{c^2}\right),
    \end{align}
    for all values of $k$ such that $|k|\leq \frac{c}{2}.$ Inserting this result into Eq.~\eqref{eq:chernoff_P_overline_X}, we find the bound on the concentration probability,
    \begin{align}
        {\rm Pr}\left[\left|\overline{X}\right|\geq x\right] \leq \exp\left(\frac{\alpha n}{2} - \frac{cx}{2}\right)
    \end{align}
    where we set $k=c/2$ satisfying the conditions from Lemma~\ref{lemma:boundedMGF}. 
\end{proof}
Lemma~\ref{lemma:bernstein} shows that if $X$ is exponentially concentrated with constants $\alpha>0$ and $c>0$, then an $n$-fold i.i.d.\ sum $\overline X$ is also exponentially concentrated.
However, the i.i.d.\ sum has a slightly heavier tail, as quantified by the constants $\exp(\alpha n /2)>0$ and $c/2>0$.

Finally, we show an application of these results for the case when we take expectation values of functions with respect to an exponentially concentrated probability distribution. As it turns out (see Lemma~\ref{lemma:taylorapproxexpectation}), it is possible to estimate the expectation value of a function using the Taylor approximation.
To this end, we introduce a familiy of real random variables $X_N$ with zero mean and exponential concentrated probability distribution,
\begin{align}
\label{eq:exp_concentrated_family}
    {\rm Pr}\left[\left|X\right|\geq x\right] \leq \alpha e^{-c\sqrt{N}x}.
\end{align}
The family parameter $N\in \mathbb R_{\geq 0}$ may take any values but must be unbounded. We then find:

\begin{restatable}[Taylor trick for expectation values]{lemma}{taylorapproxexpectation}
\label{lemma:taylorapproxexpectation}
    Assume $f:\mathbb R\rightarrow \mathbb C$ is whole and the derivatives in the origin satisfy the following condition,
    \begin{align}
        |f^{(n)}(0)| \leq \gamma^n,
    \end{align}
    for some constant $\gamma>0.$
    Furthermore, take $X_N$ to be a family of exponentially concentrated real random variables as in Eq.~\eqref{eq:exp_concentrated_family}, then we have asymptotically
    \begin{align}
    \label{eq:f_avg_expansion}
        \int dx p[X_n=x]f(x) = f(0) + \frac{\sigma^2}{2}f^{(2)}(0) + O\left(\left(\frac{\gamma }{\sqrt{N}}\right)^3\right),
    \end{align}
    as $N\rightarrow\infty,$ where $\sigma^2$ is the second moment of $X_N.$ 
\end{restatable}

\begin{proof}
For this first step, we can directly use Lemma~\ref{lemma:boundedmoments} to derive the following bound on the absolute moments of $X_N,$
\begin{align}
\label{eq:chi_n_bound}
    \chi_n:=\int dx |x|^n {\rm Pr}[X_N=x] \leq \frac{\alpha n!}{(c\sqrt{N})^n}.
\end{align}
Now let us move towards the second step where we expand the integral for the expectation value of $f.$
Since by assumption $f$ is whole, we can expand for any $x\in \mathbb R$
\begin{align}
    f(x) = \sum_{n\geq 0} x^n\frac{f^{(n)}(0)}{n!},
\end{align}
and insert into the integral,
\begin{align}
    \int dx\, p[X_N=x]f(x) = \sum_{n\geq 0}\int dx\, p[X_N=x]x^n\frac{f^{(n)}(0)}{n!},\label{eq:int_f_exp}
\end{align}
where switching integral and sum is allowed by the Fubini-Tonelli-theorem~\cite{Bauer2001}.
The first three terms of the sum are $f(0) + \frac{\sigma^2}{2}f^{(2)}(0)$ and note that the first moment vanishes because we centered the expansion of $f(x)$ around the mean $0$ of $X_N$'s distribution.
The remaining terms can be bounded by using Eq.~\eqref{eq:chi_n_bound},
\begin{align}
    \sum_{n\geq 3} \chi_n \frac{f^{(n)}(\mu)}{n!} \leq \alpha\sum_{n\geq 3} \left(\frac{\gamma}{c\sqrt N}\right)^n = O\left(\left(\frac{\gamma}{\sqrt N}\right)^3\right),
\end{align}
so long as $N>(\gamma / c)^2$ and the series converges. This concludes the proof of the Lemma.
\end{proof}

This concludes the probability theoretical excursion, and we move towards applying these results for the clock probability distributions.

\paragraph*{Step 3 -- Application to the aQPU.}
Our goal is to calculate the final state of the computation $\rho_M(t)$ for a time $t$ large enough such that all operations of the aQPU have been carried out with high probability.
For i.i.d.\ ticks the joint distribution of the times $T_{n-1,n}=T_n-T_{n-1}$ between adjacent ticks can be factorized,
\begin{align}
    {\rm Pr}[T_{0,1}=\tau_1,\dots,T_{n-1,n}=\tau_n] = \prod_{k=1}^n {\rm Pr}[T_1=\tau_k].
\end{align}
It is therefore tempting to also factorize the expression in Eq.~\eqref{eq:p_Gamma_n_prob}, which comes up in the integral of Eq.~\eqref{eq:rho_T_resolution_N_integral}.
The additional condition fixing the number of ticks $N(t)$ at time $t$, however, breaks the independence of the tick times.
In the limit of long times $t$, however, ${\rm Pr}[N(t)=n]\rightarrow 0$ for all values $n<M$, and only the case $N(t)=m$ remains in Eq.~\eqref{eq:rho_T_resolution_N_integral}.
We can quantify this properly by using the fact that
\begin{align}
\label{eq:P(N(t)=M)}
    {\rm Pr}[N(t)=m] &= {\rm Pr}[T_m\leq t],
\end{align}
because the clock does not tick more than $m$ times.
Furthermore, we assume that the time $T_1$ between two adjacent ticks is exponentially concentrated as follows: consider a family of clocks with unbounded accuracy $N$, but fixed mean $\tau$, where we write,
\begin{align}
\label{eq:probability_concentration_tick_distribution_assumption}
    {\rm Pr}[|T_1-\tau|\geq t] \leq \alpha \exp\left(-c\sqrt{N}\frac{t}{\tau}\right),
\end{align}
for constants $\alpha,c>0$.
This implies that for growing clock accuracy $N$, the tail vanishes exponentially quickly.
While this may seem like a strong assumption, we justify later that natural choices of clocks satisfy this behavior because the tick generating process is governed by exponential decay.
Under this assumption, we can invoke Lemma~\ref{lemma:bernstein} to bound the probability from Eq.~\eqref{eq:P(N(t)=M)},
\begin{align}
    {\rm Pr}[N(t)=m] &= 1 - {\rm Pr}[T_m\geq t] \\
    &\geq 1 - \exp\left(\frac{\alpha m}{2} - \frac{c\sqrt{N} (t-m\tau)}{2\tau}\right),\label{eq:P(N(t)=M)small}
\end{align}
which for $t\geq (m+1)\tau$ and high clock accuracy $N\geq m^2$ guarantees that the clock is in the state with exactly $m$ ticks.
Waiting for longer, e.g., $t_f= 2m\tau$ would allow for relaxing the condition $N\geq m^2$ on the accuracy to the weaker requirement $N\gg 1$, allowing for a more robust result.
For this reason, we work in the latter regime and write ${\rm Pr}[N(t)=m]=1-\varepsilon,$ where we remember that $\varepsilon= O\left(\exp(-c\sqrt{N}m/2)\right)$ as $N\gg 1.$ This also implies due to normalization $\sum_n {\rm Pr}[N(t)=n]=1$ that ${\rm Pr}[N(t)<m]=\varepsilon$ is small.
With this, we can approximate the state $\rho_M(t)$ in Eq.~\eqref{eq:rho_T_resolution_N_integral} by the term $n=m$ resulting in the following expression:
\begin{widetext}
    \begin{align}
        \rho_M(t) &= \int dt_1\cdots dt_m p[\Gamma_m]\rho_M(t|\Gamma_m)  + O(\varepsilon)\label{eq:rho_T_snake_1}\\
        &= \int d\tau_1 \cdots d\tau_m
        {\rm Pr}[T_{0,1} = \tau_1,\dots,T_{m-1,m}=\tau_m, T_m \leq t]
        \rho_M(t | \Gamma_m) + O(\varepsilon)\label{eq:rho_T_snake_2}\\
        &=\int d\tau_1 \cdots d\tau_m
        {\rm Pr}[T_{0,1} = \tau_1,\dots,T_{m-1,m}=\tau_m]
        \rho_M(t | \Gamma_m) + O(\varepsilon)\label{eq:long_appendix_equation_simplification}\\
        &= \int d\tau_m P[\tau_m] V_{a_{m-1}}(\tau_m) \left(\int d\tau_{m-1} \cdots \left(\int d\tau_1 P[\tau_1] V_{a_0}(\tau_1)\rho_M^\mathrm{init}V_{a_0}(\tau_1)^\dagger\right)\cdots \right)V_{a_{m-1}}(\tau_m)^\dagger + O(\varepsilon) \label{eq:rho_T_snake_4}
    \end{align} 
\end{widetext}
where we have used the result from Eq.~\eqref{eq:P(N(t)=M)small} in the step to Eq.~\eqref{eq:long_appendix_equation_simplification}, which is simply another instance where we use the fact that at time $t\geq t_f= 2M\tau,$ the contributions from the cases where the the aQPU finishes after time $t$ are extremely small and can be bounded by $\varepsilon.$
Between line~\eqref{eq:rho_T_snake_1} and~\eqref{eq:rho_T_snake_2} we change integration variables from tick times $(T_k,t_k)$ to the time between ticks, $(T_{k-1,k},\tau_k)$.
In the final line of our derivation here~\eqref{eq:rho_T_snake_4}, we denote by $P[\tau_k]={\rm Pr}[T_1=\tau_k]$ the probability that the time between tick $k-1$ and tick $k$ equals $\tau_k.$ It is possible to factorize the equations because the ticks are assumed to be independent.
Finally, the unitaries $V_{a_k}(\cdot)$ are those generated by the Hamiltonian corresponding to the program $\mathcal A,$
\begin{align}
    V_{a_k}(\tau_{k+1}) = \exp\left(-iH_{M,a_{k}}\tau_{k+1}\right),
\end{align}
evolved for some time $\tau_{k+1}.$
In this form, we see that approximately, the aQPU acts like a concatenation of mixed unitary channels on the memory system, and we can prove the following Prop.~\ref{prop:clockchannel}.

\begin{restatable}[Clock channel]{proposition}{clockchannel}
\label{prop:clockchannel}
    The i.i.d.\ recursion relation for exponentially concentrated ticks $T_1$ in the high accuracy limit $N$ can be approximated as
    \begin{align}
        \int dt P[t] V_{a}(t) \rho V_{a}(t)^\dagger = V_{M,a} \rho V_{M,a}^{\dagger} + O\left(\frac{\tau^2 \|H_{M,a}\|^2}{N}\right).\label{eq:clockchannel_prop_eq}
    \end{align}
    We abbreviate the probability distribution for a single tick with $P[t]$, and $V_{M,a}$ is  generated by applying the Hamiltonian $H_{M,a}$  for exactly the desired duration $\tau$ as in the main text.
\end{restatable}

\begin{proof}
The result follows immediately by using the trick from Lemma~\ref{lemma:taylorapproxexpectation}. We provide the explicit prefactors up to the second order expansion beyond the Eq.~\eqref{eq:clockchannel_prop_eq}, and notice the following points: in Lemma~\ref{lemma:taylorapproxexpectation}, the mean of the distribution was assumed to be in the origin; by shifting both the random variable as well as the function $f$, we can generalize to a non-zero mean $\tau>0$ of $P[t]$.
Our function $f$ is given by
\begin{align}
    f(t) =  V_{a}(t) \rho V_{a}(t)^\dagger,
\end{align}
and is naturally whole and the derivatives are nested commutators with the Hamiltonian such that we can generically bound
\begin{align}
    \left|f^{(n)}(\tau)\right| \leq \|H_{M,a}\|^n.
\end{align}
Inserting this into Lemma~\ref{lemma:taylorapproxexpectation}, we find the following result:
\begin{align}
     \int dt P[t] V_{a}(t) \rho V_{a}(t)^\dagger &= V_{M,a}\rho V_{M,a}^\dagger  \\
     &\hspace{-2cm} - \frac{\tau^2}{2 N}\left[H_{M,a}\left[H_{M,a},\rho\right]\right] + O\left(\frac{\tau^3 \|H_{M,a}\|^3}{N^{3/2}}\right), \nonumber
\end{align}
in the limit of high clock accuracy $N\gg \tau^2 \|H_{M,a}\|^2.$
This was all we intended to prove.
\end{proof}

We are now in a position to conclude the proof of Eq.~\eqref{eq:prec} from the main text, quantifying the overall fidelity for the computation.

\begin{proof}[Proof of Eq.~\eqref{eq:prec}]
    Using Eq.~\eqref{eq:rho_T_snake_4} together with the result from Prop.~\ref{prop:clockchannel} yields the desired statement.
    To be more precise, we can introduce a maximum angle of rotation $\phi_\mathrm{max}$ as in the main text to upper bound all the contributions from Prop.~\ref{prop:clockchannel} by one number.
    All-in-all, the leading order terms will be given by
    \begin{align}
    \label{eq:rho_M_final_approx}
        \rho_M(t)\Big|_{t\geq t_f} = V_\mathcal{A} \rho_M^\mathrm{init} V_\mathcal{A}^\dagger + O\left(\frac{L}{N}\phi_\mathrm{max}^2\right) + O(\varepsilon),
    \end{align}
    where $L\leq m$ is the number of non-trivial operations in the program.
    The contribution $O(\varepsilon)$ vanishes exponentially in $N$, hence, we can drop it and conclude the proof of the corollary.
\end{proof}

\paragraph*{Why exponential concentration?} 
After having successfully argued in the previous two steps how Prop.~\ref{prop:clockchannel} and Eq.~\eqref{eq:prec} come about, we analyze the assumptions that have lead to these results.
In particular, we want answer the following questions:
\begin{itemize}
    \item Under which conditions are the tick probability distributions of a family of clocks with unbounded accuracy $N$ exponentially concentrated as assumed in Eq.~\eqref{eq:probability_concentration_tick_distribution_assumption}?
    \item Which other examples of quantum clocks satisfy Prop.~\ref{prop:clockchannel}?
\end{itemize}
Let us start with the first item, and we remind ourselves of the assumption that we consider clocks producing i.i.d.\ ticks. Following the results in~\cite{Meier2023}, we can express the cumulative probability distribution ${\rm Pr}[T_1\geq t]$ of the time $T_1$ between subsequent ticks as follows,
\begin{align}
\label{eq:cpd}
    {\rm Pr}[T_1\geq t] = \exp\left(-\Gamma\int_0^t dt' f(t')\right),
\end{align}
where $\Gamma=\max_{\rho}\tr\left[J_C^\dagger J_C \rho \right]$ is the maximum rate of the clock's tick generating channel, and $f(t)$ some function with values in $[0,1]$ describing the conditional tick probability of the clock.
This form already reveals the main reason why we would expect exponential concentration of the clock's ticks: the tail for $t\rightarrow\infty$ of the distribution is naturally exponentially suppressed because the tick generating process is governed by exponential decay, and for $t\rightarrow-\infty,$ the distribution is bounded because by definition the tick must happen after $t\geq 0.$

For the purpose here, we are interested in a family of such clocks, where the clock's average tick time is fixed to $\tau$ but the decay rate $\Gamma$ is growing, giving a growing clock accuracy $N$ which is proportional to $\Gamma^2 /\tau^2.$
Looking at the left tails to bound ${\rm Pr}[T_1\leq \tau - t]$ for $t\geq 0,$ we notice that ${\rm Pr}[T_1\leq 0]=0$ by definition, which reduces the problem of finding an exponential envelope to an optimization problem on the compact interval $[0,\tau]$.
For the right tails of the tick probability distribution ${\rm Pr}[T_1\geq \tau + t],$ where again $t\geq 0,$ we notice that if there exists some constant $c>0$ such that $f(t+\tau)\geq c$ for all $t\geq 0,$ then we can bound ${\rm Pr}[T_1\geq \tau+t]$ from above using the exponential envelope $e^{-c\Gamma t}.$ 
The exponent can be rewritten in terms of the clock accuracy to yield $e^{-c' \sqrt{N} t/\tau},$ where the constant $c'$ is chosen to satisfy $c'\sqrt{N}=c\Gamma \tau.$
The parameter $\alpha$ in the definition for exponential probability concentration in Eq.~\eqref{eq:probability_concentration_tick_distribution_assumption} can now be chosen such that the sum of the left and right tail are bounded by $\alpha e^{-c' \sqrt{N} t/\tau},$ which may be understood as a general recipe to examine whether a family of tick probability distributions satisfies exponential concentration.

Going towards the second item, we consider an example of a clock that satisfies the expansion in Eq.~\eqref{eq:clockchannel_prop_eq} Prop.~\ref{prop:clockchannel}:
a clock based on exponential decay as used in the experimental proposal of the main text.
Its tick distribution is ${\rm Pr}[T_1\geq t]=e^{-\Gamma t}$.
A modification which improves this clock's accuracy is concatenating several ticks and say have only every $k$th decay count as an actual tick, but to ensure that the average time between two ticks that we count stays the same, we replace $\Gamma \rightarrow k\Gamma$.
The resulting probability distribution is given by the sum of $k$ i.i.d.\ exponential random variables with same rate $k\Gamma$, which is also known as the Erlang distribution~\cite{Cox1962},
\begin{align}
\label{eq:scaled_erlangian}
    {\rm Pr}[T_{(k)}\geq t] &= \sum_{n=0}^{k-1}\frac{(k\Gamma t)^n}{n!}\exp\left(-k\Gamma t\right).
\end{align}
A clock with this cumulative tick distribution has accuracy $N=k$ at average tick time $\langle T_{(k)}\rangle=\tau=1/\Gamma.$
The random variable $T_1$ is exponentially concentrated by definition with parameters $\alpha=2$ and $c=\Gamma$.
We can thus use Lemma~\ref{lemma:bernstein} with a small change (see, e.g., Thm.~1.13 in~\cite{Rigollet2015}) to account for the rescaling $\Gamma\rightarrow k\Gamma$ to yield
\begin{align}
    {\rm Pr}\left[\left|T_{(k)}-\tau\right|\geq t\right] \leq 2\exp\left(-\frac{N}{2}\mathrm{min}\left\{\frac{t}{\tau},\frac{t^2}{\tau^2}\right\}\right),
    \label{eq:PrErlang}
\end{align}
with $N=k$ the accuracy.
Since the tails of Eq.~\eqref{eq:PrErlang} are asymptotically bounded by an exponential envelope, the expansion in Prop.~\ref{prop:clockchannel} holds for the first two orders.

For an illustration how this clock can be used for the aQPU, we refer the reader to Sec.~\ref{sec:numerical_example_bellstate} and Figs.~\ref{fig:evolution} and~\ref{fig:fidelity} where an example is simulated numerically.

\subsection{\label{sec:termodynamic_cost_of_computation}Thermodynamic cost of computation}
Circuit-based quantum computation fundamentally relies on clocks for the timing of the computation, and timekeeping comes at an entropic cost~\cite{Milburn2020,Erker2017}.
Still there are potential additional contributions during initialization and in the final readout that we discuss in the following.
Our work captures primarily the cost during the computation due to the clock, as well as the clock for initialization.
The readout, which is irreversible due to the measurement, is a problem of independent interest as addressed in~\cite{Guryanova2020,Schwarzhans2023}.

\paragraph*{Clockwork entropy production.}
We will start with investigating the cost for timekeeping during the computation.
The evolution of the aQPU constitutes three processes: the evolution of the clockwork, the stochastic ticking and the clock-instruction-computer interaction, as one can identify from Eq.~\eqref{eq:L_aQPU}.
Hermitian evolution generators like the clock-instruction-computer interaction do not come at an energetic cost because they are time-independent and energy-conserving. 
The open system's component in the clockwork and the tick generation on the other hand are responsible for the unidirectional evolution of the aQPU and therefore, are expected to produce entropy. Examining this entropy production we begin with the clockwork which formally, corresponds to the contribution from $\mathcal L_C$ (see Eq.~\eqref{eq:dot_rho_CR}) in $\mathcal L_\mathrm{aQPU}.$ The entropy-production in the regime that we consider in the main text, i.e., where there is no conjugate tick generation operator $\overline{J}\propto J ^\dagger,$ is not well-defined~\cite{Silva2023}.
Including the entropy production of the conjugate ticking process requires that we allow for a non-zero probability of operation $n-1$ being carried out after operation $n$ due to detailed balance~\cite{Silva2023}.
While this does not change anything fundamentally, the mathematics required to deal with this introduce additional challenges which we address in Appendix~\ref{appendix:generalization_for_backwards_ticking_clocks}.

The entropy production we define here follows from a set of assumptions, which clarify how energy-changes in the system relate to heat dissipation and work.
One of these assumptions is that the system-bath interactions are energy-preserving, in that case the entropy production can be identified unambiguously in the separation as proposed by Spohn~\cite{Spohn1978,landi_paternostro_rev}.
Furthermore, the jump operators $L_\ell$ appearing in $\mathcal L_C$ must satisfy a property called \textit{local detailed balance}~\cite{Seifert2012,Maes2021,landi_paternostro_rev}.
This means that for every operator $L_\ell$ there exists conjugate operator $\overline L_\ell$, which is proportional to $L_\ell^\dagger,$ with the constant of proportionality related to the entropy production $\Delta \sigma_\ell$ per unit population that undergoes the jump $L_\ell.$
In the regime of local detailed balance jump operators always come in pairs, which we can denote $(L_\ell,\overline L_\ell)$ and that are related via
\begin{align}
\label{eq:L_ell_detailed_balance}
    \overline{L}_\ell = e^{-\Delta \sigma / 2}L_\ell^\dagger.
\end{align}
On this level of specificity the entropy production can be related to the thermodynamic entropy $\Delta \sigma=\beta \Delta q$ given by the product of the inverse temperature of the bath $\beta$ generating the transition and heat exchange $\Delta q$ in the jump process~\cite{Hofer2017,Carmichael2002,Walls2008}. With local detailed balance, we are in a position to quantify how much entropy the clockwork of the aQPU produces.
Each unit population undergoing the jump $L_\ell$ produces $\Delta \sigma_\ell$ entropy, and each reverse jump $\overline L_\ell$ produces $-\Delta \sigma_\ell$ entropy.
When working with the master equation evolution we have an ensemble average, which allows us to calculate an average entropy production rate of the clockwork by weighting the probability currents that jump through $L_\ell$ and $\overline L_\ell$, such that~\cite{Spohn1978,landi_paternostro_rev}
\begin{align}
    \left\langle \dot \Sigma_\mathrm{cw}(t)\right\rangle = \sum_{\ell}\Delta \sigma_\ell  \, \tr\left[\left(L_\ell^\dagger L_\ell - \overline L_\ell^\dagger \overline  L_\ell\right) \rho_{C}(t)\right].
\end{align}
Based on this expression the total average entropy production at time $t$ can be calculated using
\begin{align}
\label{eq:Sigma_CW_T}
    \left\langle\Sigma_\mathrm{cw}(t)\right\rangle = \int_0^t dt' \left\langle \dot \Sigma_\mathrm{cw}(t')\right\rangle.
\end{align}
In this case, the total entropy production can be resolved into contributions coming from each of the ticks separately.

This separation is especially useful in the case that the entropy per tick, denoted by $\Sigma_{\mathrm{tick}},$ is the same for all ticks.
In this case, the relationship between the program fidelity and the entropy production per tick takes a particularly concise form, because the integrated entropy production of the clockwork factorizes as
\begin{align}
    \langle \Sigma_\mathrm{cw}(t)\rangle = \langle N(t)\rangle \Sigma_{\mathrm{tick}},
\end{align}
where $\langle N(t)\rangle$ is the expected number of ticks at time $t$.

\paragraph*{Thermodynamic cost of initialization.}
We now consider contributions from initialization.
For the initial state, tick and instruction register as well as memory system need to be prepared.
Modeling the clockwork as an autonomous thermal machine as we have, its state initially does not require any particular preparation since it will tend toward a steady-state.
For its tick register, however, it is imperative that it is initially prepared in the $\ket{0}_{T}$ state.
Similarly, the punch card state in the instruction register should encode a program $\ket{\mathcal A}_{I}$ and the memory system should be in the initial state of the computation, conventionally labeled as $\ket{0}_M$.
In our setting, all three systems should ideally be in a pure state, which are impossible to achieve with only finite thermodynamic costs~\cite{Taranto2023}.
One could also imagine a generalized model where the punch card states and tick register states are not pure but mixed states.
This would not resolve the problem of preparation, because the preparation of mutually orthogonal mixed states also comes at a non-zero cost.
Working with the pure states of our model, we can only approximate the initial states $\ket{0}_{T}\otimes\ket{\mathcal A}_I\otimes\ket{0}_M$ to finite accuracy and therefore, the program fidelity $\mathcal F_{\mathcal A}$ would obtain an additional error term.
If we work in the paradigm that initially, we only have access to thermal states at inverse temperature $\beta>0$ then we can investigate the cost of using these resources we have access to freely to improve some subset of them.  For example, one could adjust the setup of the aQPU such that it is possible  to use the pure-state-preparation protocol from~\cite{Reeb2014}, which has been used~\cite{Meier2023a,Xuereb2023a} to derive a relationship between entropy production $\Sigma_\mathrm{init}$ for state-preparation and the fidelity of the preparation.
Let us define the fidelity of the initial state preparation as $\varepsilon>0$ by 
\begin{align}
    \varepsilon := 1-\left\langle 0_{T},\mathcal A_I, 0_M\big|\tau[\beta]_{TIM}\big|0_{T},\mathcal A_I, 0_M\right\rangle,
\end{align}
where $\tau[\beta]_{TIM}$ is the initial thermal state of tick register ($T$), instruction register ($I$) and memory ($M$).
If we assume an preparation protocol in $L_{\rm init}$ steps, then we can relate the entropy production $\Sigma_\mathrm{init}$ in the thermal baths used for the preparation to $\varepsilon$ and $L_{\rm init}$ by using results from~\cite{Meier2023a,Xuereb2023a},
\begin{align}
    \varepsilon &\leq \frac{1}{\Sigma_\mathrm{init}}\exp(-L_{\rm init}\Sigma_\mathrm{init} + 2 W),
\end{align}
where $W$ is a measure for the number of qubit equivalents for the registers and memory $TIM$ that we prepare, and is given by $W=\log(d-1),$ with $d$ being the Hilbert space dimension of $TIM$.
Going back to the fidelity $\mathcal F_{\mathcal A},$ the error $\varepsilon$ from the initialization will simply carry over to the final state and provide an upper bound.
Thus, we can modify the relationship in Eq.~\eqref{eq:prec} to give
\begin{align}
    \mathcal F_{\mathcal A} \leq 1 - O\left(\frac{M}{f(\Sigma_{\mathrm{tick}})}\phi_\mathrm{max}^2\right) - \frac{e^{-L_{\rm init}\Sigma_\mathrm{init} + 2 W}}{\Sigma_\mathrm{init}}.
\end{align}
This expression gives us a relationship between the fidelity $\mathcal F_{\mathcal A}$, the complexity $M\phi^2_\mathrm{max}$ and the thermodynamic cost of the aQPU.
The thermodynamic cost splits into a contribution from the clock, $\Sigma_{\mathrm{tick}}$, and a contribution, $\Sigma_\text{init}$, from the purification of the initial state of dimension $W$. 
Such an expression edifies the fact that accurate quantum computation comes at a cost, which is greater for more \textit{complex} computations.

\subsection{\label{sec:speed_vs_fidelity}Speed of computation and fidelity}
The time-energy uncertainty relation introduced by Mandelstam and Tamm~\cite{Mandelstam1991} was shown by Margolus and Levitin~\cite{Margolus_1998} to bound the rate at which a quantum state can evolve and traverse its state space.
This led to the development of the field of \emph{Quantum Speed Limits}~\cite{Deffner_2017} which bound the speed at which quantum logic gates can be executed sequentially~\cite{Lloyd_2000,levitin_toffoli_09}, but not collectively~\cite{Jordan2017}.
Here, we consider how the irreversible thermodynamic cost of computation provides a dissipative speed limit for how fast the aQPU can execute a program, complementing the fundamental quantum speed limit bounds.

\paragraph*{Optimal gate compilation.}
In this paragraph, we investigate the speed of executing an arbitrary unitary on the aQPU given access to a finite set of Hamiltonians.
Here, the notion of speed is captured by the length of the program required to approximate the desired operation.
This approximation, known as compilation, is possible due to the Solovay-Kitaev theorem~\cite{Kit97, dawson2006solovay}.
We show that while longer approximations are more accurate, they are more susceptible to physical error on the aQPU.

Since the aQPU is only able to carry out a finite number of Hamiltonians for a finite set of possible target durations it can also only generate a finite number of (non-tunable) unitaries.
This finite number of unitaries must at times be executed as products to approximate a unitary which is outside of the gate set.
In other words, a program expressed in an arbitrary gate set must be \textit{compiled} so that it may be executed on the aQPU using the gate set it has access to.
The Solovay-Kitaev Theorem~\cite{Kit97, dawson2006solovay} states that with access to a universal gate set $\mathcal V$ for $\mathrm{SU}(2)$, i.e., a set of gates closed under inversion that generate a dense subgroup of $\mathrm{SU}(2)$ we can approximate any unitary $U \in \mathrm{SU}(2)$ in the following sense:
for arbitrarily small $\varepsilon > 0,$ the unitary $U$ can be approximated using a product of $L$ gates $V_i\in\mathcal V.$ Formally, we can write $V = V_L V_{L-1}\dots V_1,$ and the distance of $V$ to the desired unitary $U$ is bounded by
\begin{align}
    \|U - V\|_\infty \leq \varepsilon 
\end{align}
where $L = O\left(\log^c(1/\varepsilon)\right)$ for some constant $c>0$.

This means that for a given unitary $U$, an approximation of a suitable accuracy $\varepsilon$ can be found at the cost of increasing the length of the approximation $L$.
This in turn also implies that the time of the compilation scales as $L\tau$ if we reduce the error $\varepsilon$ of the approximation, where $\tau$ is the average duration required by the aQPU to carry out a single gate.
Because the aQPU is a thermodynamic machine that only executes programs perfectly with access to diverging resources to power perfect clocks and run perfectly pure punch card states, we would expect that arbitrarily increasing the length $L$ can not indefinitely increase the quality of the gate compilation.
On the contrary, as we have seen in Sec.~\ref{sec:error_propagation_nonideal_clocks},
longer programs are more susceptible to error from finite resources such as imperfect timekeeping which seems to be at odds with the Solovay-Kitaev Theorem which requires longer approximations for higher accuracy.
This motivates us to ask \textit{what  resources  suffice  for the aQPU to execute the shortest program with the highest accuracy?} In other words, what is the cost of a program being run on the aQPU at a given speed and accuracy?
\begin{restatable}[Compiling with finite resources]{proposition}{aqpusolovay}
\label{prop:aqpusolovay}
An aQPU featuring a master clock generating exponentially concentrated i.i.d.\ ticks in the limit of high accuracy $N\gg 1$ and access to a finite set of Hamiltonians which generate a universal gate set $\mathcal V$ for $\mathrm{SU}(2)$, can approximate any $U \in \mathrm{SU}(2)$ using a program $\mathcal A=(a_0,\dots,a_{L-1})$ of $L$ unitaries $V_\mathcal{A} = V_{M,a_{L-1}}\dots V_{M,a_0}$ with error 
\begin{align}
\label{eq:solovay_eq_aqpu}
    \left\|U\ketbra{0}{0}_M U^\dagger - \rho_M^{\mathcal A}\right\|_1
    \leq \exp\left(-\Omega\left(L^{1/c}\right)\right) +  {O}\left(\frac{L}{N}\phi_\mathrm{max}^2\right) 
    \end{align}
where $c$ is a constant and the big-$\Omega$ notation as defined by Knuth~\cite{Knuth1976}.
The state $\rho_M^{\mathcal A}$ is the memory system's state for the aQPU with program $\mathcal A$ evaluated at time $t\geq 2 L\tau$.
\end{restatable}

\begin{proof}
Let the aQPU have access to a finite set $\mathcal V$ of $K$ Hamiltonians which generate a dense subset of unitaries $\langle \mathcal V \rangle$ in $\mathrm{SU}(2).$
Furthermore, we look at an aQPU with clock accuracy $N$ satisfying the fidelity relationship from Eq.~\eqref{eq:prec} given a program length $L\equiv m$.
If we want to execute a unitary $U\in\mathrm{SU}(2)$ on a qubit having access only to the elementary gates in $\mathcal V,$ the 
Solovay-Kitaev Theorem implies that this is possible approximately with an error $\varepsilon>0$ using a product $V_{\mathcal A} = V_{M,a_{L-1}} \cdots V_{M,a_0}$, where $V_{M,a_\ell} \in \mathcal V$ and $\mathcal A = (a_0,\dots a_{L-1})$ is the program of the aQPU.
The error is quantified such that the distance between target unitary $U$ and the compiled program $V_{\mathcal A}$ is bounded by $|| V_\mathcal A - U ||_\infty \leq \varepsilon$ where the error $\varepsilon$ and the length $L$ of the program are related via $L = O(\log^c(1/\varepsilon))$ for some constant $c>0$ that is independent of the desired gate.

The quantity we actually want to calculate is the distance between $U\ket{0}_M$ and the state $\rho_M^{\mathcal A}$ that the aQPU with imperfect clock can generate.
To this end, we consider the trace distance between the target unitary $U$ applied to an initial state $\ket{0}_M$ and the approximation executed on the aQPU
\begin{align}
    &\left\| U\ketbra{0}{0}U^\dagger - \rho_M^{\mathcal A}\right\|_1 \leq
    \left\| U\ketbra{0}{0}U^\dagger - V_{\mathcal A}\ketbra{0}{0}V_{\mathcal A}\dagger\right\|_1  \nonumber \\
    &\qquad + \left\| V_{\mathcal A}\ketbra{0}{0} V_{\mathcal A}^\dagger - \rho_M^{\mathcal A}\right\|_1, \label{eq:triangle} 
\end{align}
where we have used a midpoint trick and a triangle inequality to split this distance into the first term which captures the quality of the approximation and a second term which captures how well the aQPU exectues the approximation program.
The trace-distance used here is defined as
\begin{align}
    \|\rho-\sigma\|_1 := \frac{1}{2}\tr\left[\sqrt{\left(\rho-\sigma\right)^\dagger\left(\rho-\sigma\right)}\right].
\end{align}
Focusing on the second term, we may readily apply Prop.~\ref{prop:clockchannel} to obtain
\begin{align}
    \left\| V_{\mathcal A}\ketbra{0}{0}_M V_{\mathcal A}^\dagger - \rho_M^{\mathcal{A}}\right\|_1 = O\left(\frac{L}{N}\phi_\mathrm{max}^2\right)\label{eq:dist_clock} 
\end{align}
since $\rho_M^{\mathcal A} = V_{\mathcal A}\ketbra{0}{0}_MV_{\mathcal A}^\dagger + O\left(\frac{L}{N}\phi_\mathrm{max}^2\right)$ by  Eq.~\eqref{eq:rho_M_final_approx}.
This leaves the first term which we bound by considering that the approximation is guaranteed to satisfy $||U - V_{\mathcal A}||_\infty\leq \varepsilon$ by the Solovay-Kitaev Theorem~\cite{dawson2006solovay}. By definition of the Schatten $\infty$-norm we have 
\begin{align}
   \varepsilon \geq ||U - V_{\mathcal A}||_\infty &= \sup_{\ket{\psi} \in \mathcal{H}} \left\|(U - V_{\mathcal A})\ket{\psi}_M\right\|_2 \\
    &= \sup_{\ket{\psi} \in \mathcal{H}} \sqrt{\bra{\psi}(U^\dagger - V_{\mathcal A}^\dagger)(U - V_{\mathcal A})\ket{\psi}_M}. \nonumber
\end{align}
Squaring both sides and making use of the definition of the supremum we obtain
\begin{align}
    \varepsilon^2 &\geq \sup_{\ket{\psi} \in \mathcal{H}}\bra{\psi}(U^\dagger - V_{\mathcal A}^\dagger)(U - V_{\mathcal A})\ket{\psi}_M \nonumber \\
    &= \bra{0}(U^\dagger U + V_{\mathcal A}^\dagger V_{\mathcal A} - U^\dagger U_{\mathcal A} - V_{\mathcal A}^\dagger U)\ket{0}_M \nonumber \\
    &= 2\left(1 -  \mathfrak{Re}\left\{\bra{0}U^\dagger V_{\mathcal A} \ket{0}_M\right\} \right) \nonumber \\
    &\geq 2\left(1 -  \left|\bra{0}U^\dagger V_{\mathcal A} \ket{0}_M\right|\right),\label{eq:e^2_geq_1-something}
\end{align}
where we have used $\mathfrak{Re}\{z\} \leq |z|,\, \forall z \in \mathbb{C}$ in the last inequality. By multiplying and dividing Eq.~\eqref{eq:e^2_geq_1-something} by the same factor $1 +  \left|\bra{0}U^\dagger V_{\mathcal A} \ket{0}_M\right|,$ we get
\begin{align}
    \varepsilon^2 \geq 2\left(\frac{1 -  \left|\bra{0}U^\dagger V_{\mathcal A} \ket{0}_M\right|^2}{1 +  \left|\bra{0}U^\dagger V_{\mathcal A} \ket{0}_M\right|}\right).
\end{align}
Since states are normalized, Cauchy-Schwarz implies that $1 \geq \left|\bra{0}U^\dagger V_{\mathcal A} \ket{0}_M\right|$ and we can further bound 
\begin{align}
    \varepsilon^2 \geq 1 -  \left|\bra{0}U^\dagger V_{\mathcal A} \ket{0}_M\right|^2, \label{eq:eps_ineq}
\end{align}
which will allow us to complete our proof. The trace distance of two pure states $\ket{\psi},\, \ket{\phi}$ simplifies to $||\ketbra{\psi}{\psi} - \ketbra{\phi}{\phi}||_1 = \sqrt{1 - |\braket{ \psi | \phi }|^2}$ and so the first contribution to Eq.~\eqref{eq:triangle} reduces to 
\begin{align}
    \left\|U\ketbra{0}{0}_M U^\dagger - V_{\mathcal A}\ketbra{0}{0}_M V_{\mathcal A}^\dagger\right\|_1 = \sqrt{1 - \left|\bra{0}U^\dagger V_{\mathcal A}\ket{0}_M\right|^2}    
\end{align}
and so by Eq.~\eqref{eq:eps_ineq} we can bound the contribution by 
\begin{align}
    \left\|U\ketbra{0}{0}_M U^\dagger - V_{\mathcal A}\ketbra{0}{0}_M V_{\mathcal A}^\dagger\right\|_1 \leq \varepsilon,\label{eq:dist_han-solo}
\end{align}
giving the desired result as a combination of Eqs.~\eqref{eq:dist_clock} and~\eqref{eq:dist_han-solo},
\begin{align}
   \left\| U\ketbra{0}{0}_M U^\dagger - \rho_M^{\mathcal A}\right\|_1 \leq \varepsilon +  O\left(\frac{L}{N}\phi_\mathrm{max}^2\right),
\end{align}
a modified Solovay-Kitaev theorem. This shows that whilst any unitary can be approximated with an error $\varepsilon$ that scales inversely in the length $L$ of the approximation, finite thermodynamic resources such as entropy production introduce errors which scale with the length of the approximation. The Solovay-Kitaev construction gives $L = {O}\left(\log^c(1/\varepsilon)\right)$ as the inverse scaling relationship and can derived from the inequality given in~\cite{Ozo09, Nielsen2010} 
\begin{align}
    L < \frac{5}{4}\left(\frac{\log(1/C^2\varepsilon)}{\log(1/C\varepsilon_0)}\right)^c L_0,
\end{align}
where $L_0$ is the length of an initial guess approximation with error $\varepsilon_0$ from which Solovay-Kitaev algorithm starts, $c$ is a bounded constant depending on the choice of compiling algorithm but independent of the desired gate and $C$ is an error scaling constant. Inverting this inequality one obtains 
\begin{align}
    \varepsilon = \exp\left(-\Omega\left(L^{1/c}\right)\right),
\end{align}
in the limit of large $L.$ Here, we have used the Knuth definition of the big-$\Omega$ notation~\cite{Knuth1976}.
Thus, the error $\varepsilon$ in the gate approximation scales as a stretched exponential in $L$.
The modified Solovay-Kitaev theorem can now be stated as 
\begin{align}
    \left\|U\ketbra{0}{0}_M U^\dagger - \rho_M^{\mathcal A}\right\|_1
    \leq \exp\left(-\Omega\left(L^{1/c}\right)\right) +  {O}\left(\frac{L}{N}\phi_\mathrm{max}^2\right), \nonumber
\end{align}
where the first contribution corresponding to the error $\varepsilon$ in the approximation decays as a stretched exponential with the length $L$ of the approximation. The second contribution grows linearly in $L$ because the errors from the aQPU approximation of the program $\mathcal A$ add up the longer the program and thus this error scales linearly with $L$, which is all we wanted to show.
\end{proof}

Whilst the current result captures the impact on any compiled computation on qubits this result can be extended to any unitary in $\mathrm{SU}(D)$ with minor alterations using generalizations of the Solovay-Kitaev theorem which can be found in~\cite{dawson2006solovay} and are out of scope for this work.

\paragraph*{Clock speed vs. fidelity trade-off.}
Let us now consider a fixed program $\mathcal A$ to be executed on the aQPU that is encoded by a unitary $V_\mathcal{A}$ and examine how it is impacted by being executed at different clock speeds at different clock accuracies.
Using the results from Sec.~\ref{sec:error_propagation_nonideal_clocks}, we know that the aQPU with clock in the high but finite accuracy regime can approximate $V_\mathcal{A}$ only up to an error that scales inverse linear in the clock accuracy $N$.

One naive trick we could play to improve the clock's accuracy, whilst keeping the ticking rate the same,  is the following: we change the  interaction Hamiltonian, such that only every second tick of the clock, the instruction performed is switched, and each of the generating Hamiltonians from the instruction set is re-scaled by a factor of $1/2,$ to compensate for the increase in time between when the instructions are switched. For i.i.d.\ ticks, this increases the clock accuracy by a factor of $2$ at the cost of being twice as slow.
More generally, if we switch instructions only every $k$th tick, we transform
\begin{align}
\label{eq:speed_fidelity_coarse_graining_prescription}
    \tau \mapsto k{\tau},\quad H_{M,a}\mapsto \frac{H_{M,a}}{k},\quad N\mapsto kN.
\end{align}
The increase in accuracy comes from the fact that summing $k$ i.i.d.\ random variables increases the average as $\tau\mapsto k\tau$ but the standard deviation only as $\sigma\mapsto\sqrt{k}\sigma$.
Since the accuracy $N$ is the ratio $\langle T_1\rangle^2 /\mathrm{Var}[T_1],$ we find that $N\mapsto kN.$
We find that with the increase in accuracy, the program fidelity $\mathcal {F_A}$ also increases,
\begin{align}
    \mathcal{F_A} = 1 - O\left(\frac{L}{kN}\phi_\mathrm{max}^2\right).
\end{align}
On the other hand, however, the computational time also increases by a factor of $k$, due to the increase of the average time between two instructions from $\tau$ to $k\tau$.

\subsection{\label{sec:numerical_example_bellstate}Numerical example: Bell-state preparation}
The deterministic generation of Bell states has been a fundamental benchmarking task in quantum computation both due to its experimental simplicity (in comparison to tasks of higher complexity) and relationship to quantum advantage in tasks such as quantum cryptography~\cite{ekert,Yin2020}, metrology~\cite{metrology} and communication~\cite{Buhrman2016}. Below we give a numerical simulation of a minimal aQPU generating a two qubit Bell state to give a sense for how the relationships found in the results section manifest themselves in the dynamics and to investigate what it would take to achieve fidelities on par with those obtained in experiment, on the aQPU. Our minimal aQPU consists of a two qubit computational memory register, a qutrit tick register to time the three steps of the computation, a three qutrit punch card state to codify the three different gate positions and lastly a simple quantum clock based on the phenomenon of exponential decay~\cite{Erker2017} to time the operations. The simulation of this minimal aQPU is publicly available at~\cite{Meier2024Git} and we give a brief overview of this simulation below.

\begin{figure}
    \centering
    \includegraphics[width=\columnwidth]{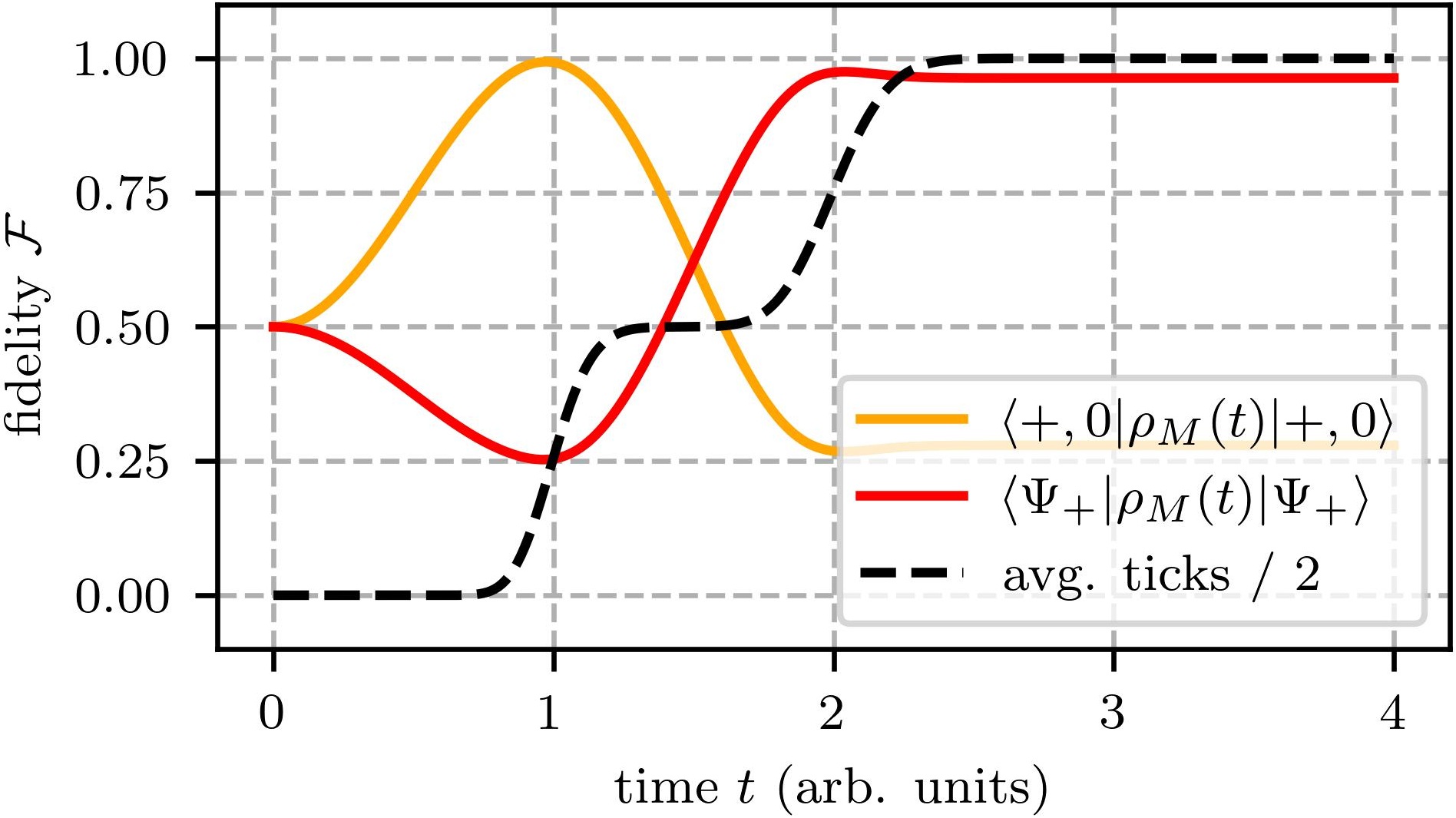}
    \caption{This figure shows the aQPU evolution for a programmed Hadamard and CNOT gate applied to an initial state $\ket{00}$ on two qubits.
    The clock used is a coarse-grained exponential decay clock with accuracy $N=80.$
    We see how the fidelity of the computational state $\rho_M(t)$ with the $\ket{+,0}$ state grows to almost unity during the time $[0,\tau]$ of the first tick (orange curve), and then how the fidelity with the Bell-state $\ket{\Psi_+}=\frac{1}{\sqrt{2}}\left(\ket{00}+\ket{11}\right)$ grows to almost unity after the second tick at $2\tau,$ albeit with larger error (red curve).
    The slight decrease in fidelity of the red curve after $2\tau$ comes from the fact that the tick expectation value (black dashed curve) is not yet exactly $2.$ This corresponds to a small possibility that the clock takes longer to achieve the third tick than expected, leading to the CNOT Hamiltonian running longer than expected, resulting in error.}
    \label{fig:evolution}
\end{figure}

\paragraph*{Bell-state generation program.}
The way we generate the Bell-pair is by starting with the computational memory system in the $\ket{00}_M$ state initially and encoding a sequence of the operations
\begin{enumerate}
    \item \verb|HADAMARD| on qubit 1,
    \item \verb|CNOT| on qubit 1 and 2.
\end{enumerate} 
We achieve this using a gate set 
\begin{equation}
    \mathcal V = \left\{U_{\verb|H|}\otimes\mathds 1, U_{\verb|CNOT|}\right\}
\end{equation}
where $U_{\verb|H|} = \ketbra{+}{0}+\ketbra{-}{1}$ is the Hadamard and $U_{\verb|CNOT|}=\ketbra{0}{0}\otimes\mathds 1 + \ketbra{1}{1}\otimes X$ the CNOT.
Both operations are written in the computational basis $\{\ket{0},\ket{1}\}$, with $\ket{\pm}=\frac{1}{\sqrt{2}}(\ket{0}\pm\ket{1})$, and $X$ the Pauli-$X$ gate.
These two gates are generated by the Hamiltonians 
\begin{equation}
    H_{\verb|H|} = -\frac{\pi}{2}\left(\mathds 1_2 - U_{\verb|H|}\right),
\end{equation}
for the Hadamard gate $U_{\verb|H|},$ and
\begin{equation}
    H_{\verb|CNOT|} = \frac{\pi}{2} \begin{bmatrix}
0 & 0\\ 
0 & 1
\end{bmatrix}\otimes
\begin{bmatrix}
1 & -1\\ 
-1 & 1
\end{bmatrix}.
\end{equation}
To be explicit, the Hamiltonians are related to the gates via $U_{\verb|H|} = \exp\left(-iH_{\verb|H|}\right)$ and $U_{\verb|CNOT|}=\exp\left(-iH_{\verb|CNOT|}\right).$
That is, the aQPU must now trigger interactions in the correct order for the correct duration using a quantum clock and instruction register.

\paragraph*{Master clock.}
As for the master clock, we work with the example given in~\cite{Erker2017}.
Physically, this quantum clock is made-up of a $D$-dimensional ladder and a two-qubit heat engine.
The engine is coupled to two out-of-equilibrium heat baths that drive up the ladder's state until it decays and generates a tick.
The dynamics of this clock are well approximated by what is called the \textit{Erlang-clock}, also discussed as an example at the end of Sec.~\ref{sec:error_propagation_nonideal_clocks}, see Eq.~\eqref{eq:scaled_erlangian}.
This clock is run by a sequence of exponential decays of the same rate $\Gamma$ such that only every $D$th decay is counted as an actual tick of the clock.
Modeling this clock requires a clockwork of dimension $D,$ with internal state-space given by $\ket{0}_C,\dots,\ket{D-1}_C.$
As for the evolution, we have jump processes $\ket{k}_C\rightarrow \ket{k+1}_C$ that can be described with a Lindblad jump operator
\begin{align}
\label{eq:LC_erlang}
    L_C = D\sum_{k=0}^{D-2} \ketbra{k+1}{k}_C.
\end{align}
Here, we have chosen unitless time, and a jump rate that increases with the dimension $D$ of the clockwork.
The operator $L_C$ describes the incoherent jumps from $\ket{0}_C$ all the way up to $\ket{D-1}_C$.
Given that the clock initially starts in $\ketbra{0}{0}_C,$ this evolution is entirely classical, which simplifies the computational resources required for the numerical analysis.
The jump operator~\eqref{eq:LC_erlang} scales with $D$ to ensure that the average time until one clock-cycle is completed is always $1,$ regardless of how large $D$ is. For larger dimension $D,$ the clock has to jump through more levels, but the larger the prefactor, the faster the clock cycles through these levels.

The clock ticks are generated by the jump $\ket{D-1}_C\rightarrow \ket{0}_C$, mediated by the operator,
\begin{align}
\label{eq:JCR_erlang}
    J_{CT} = D \ketbra{0}{D-1}_C \otimes \sum_{j=0}^1\ketbra{j+1}{j}_{T}.
\end{align}
One can verify using~\eqref{eq:scaled_erlangian} that for the clock whose evolution is defined by the operators in Eqs.~\eqref{eq:LC_erlang} and~\eqref{eq:JCR_erlang},
the average time between two ticks is given by $\langle T_1\rangle \equiv \tau = 1$ and the variance by $\mathrm{Var}[T_1] = \frac{1}{D}.$ Hence, the clock accuracy here equals $N=D,$ and we would expect for larger clockwork dimensions $D,$ the computational fidelity to increase.

Finally, the interaction Hamiltonian can be written out in detail,
\begin{widetext}
    \begin{align}
    H_{\rm int} = &\ketbra{0}{0}_{T}\otimes\big(\ketbra{\text{\ttfamily H}}{\text{\ttfamily H}}_{I_1}\otimes \mathds 1_{I_2I_3}\otimes H_{\text{\ttfamily H}}+\ketbra{{\text{\ttfamily C}}}{{\text{\ttfamily C}}}_{I_1}\otimes \mathds 1_{I_2I_3}\otimes H_{\text{\ttfamily CNOT}}\big) \nonumber \\
     + &\ketbra{1}{1}_{T}\otimes\big(\mathds 1_{I_1}\otimes \ketbra{\text{\ttfamily H}}{\text{\ttfamily H}}_{I_2}\otimes \mathds 1_{I_3}\otimes H_{\text{\ttfamily H}}+\mathds 1_{I_1}\otimes \ketbra{{\text{\ttfamily C}}}{{\text{\ttfamily C}}}_{I_2}\otimes \mathds 1_{I_3}\otimes H_{\text{\ttfamily H}}\big).
\end{align}
\end{widetext}
After the second tick, the tick register ends in the subspace spanned by $\ket{2}_{T},$ where $H_{\rm int}$ does not have any support, hence the clock is idling after that final jump.
To encode the program for the Bell-state generation, we will now use the following instruction state,
\begin{align}
    \ket{\mathcal A}_I = \ket{\text{\ttfamily H}, \text{\ttfamily C}, 0}_I,
\end{align}
corresponding to the desired gate-sequence Hadamard, $U_{\verb|H|}\otimes \mathds 1_2$ and then CNOT, $U_{\verb|CNOT|}.$
Applied to the initial state $\ket{0}_M,$ we would get out the state
\begin{align}
    \ket{00}_M \mapsto V_{\mathcal A} \ket{00}_M = \frac{1}{\sqrt{2}}\left(\ket{00}_M+\ket{11}_M\right) = \ket{\Psi_+}_M.
\end{align}
The goal of the aQPU evolution given by the Lindbladian $\mathcal L_\mathrm{aQPU}$ is to approximate this state.
By starting with the inital state
\begin{align}
    \ket{\Psi_\mathrm{init}} = \ket{0}_C\otimes\ket{0}_{T} \otimes\ket{\mathcal A}_I \otimes \ket{00}_M,
\end{align}
and evolving the aQPU Lindbladian for some time $t=4$, that is twice the time we would expect to be required for completing the Hadamard and CNOT gate, we get
\begin{align}
    \rho_M^\mathcal{A}(t) = \tr_{CHR}\left[e^{t\mathcal{L}_\mathrm{aQPU}} \ketbra{\Psi_\mathrm{init}}{\Psi_\mathrm{init}}\right].
\end{align}
Ideally, we would find that $\rho_M^{\mathcal A}(4)$ is approximately $\ket{\Psi_+}_M$, and quantitatively, this corresponds to minimizing the error
\begin{align}
    1-\mathcal F_{\mathcal A}=1-\braket{\Psi_+ |\rho_T^{\mathcal A}(4)|\Psi_+}.
\end{align}
One may convince themselves using the final example in Appendix~\ref{sec:error_propagation_nonideal_clocks} that the family of clocks at hand, parametrized by $D\in\mathbb N$, satisfies all the assumptions of Eq.~\eqref{eq:prec}, which shows that as $D\rightarrow\infty,$ the program fidelity $\mathcal F_\mathcal{A}$ approaches 1 asymptotically with leading error $1/D.$ This scaling is also verified by the simulation results presented in Fig.~\ref{fig:fidelity}.

In recent experiment~\cite{bell_state_experiment} Bell states with an infidelity of $10^{-4}$ were generated using trapped calcium ions.
To achieve such a fidelity for deterministic Bell state generation using the minimal aQPU we have modeled, we can extrapolate the linear relationship it exhibits between clock accuracy and infidelity and conclude that a clock accuracy on the order of $10^4$ is required. Physically, this would imply that the aQPU is powered by a clock which ticks 10,000 times before being off by one tick.
If as a minimal model for the clock we chose one of the proposals in~\cite{Erker2017,Barato2016}, the entropy production for the entire computation would also be at least of the order $10^4\times k_B,$ using the entropy curve plotted in Fig.~\ref{fig:fidelity}.  This number is several orders of magnitude lower than any classical control would consume to achieve the same goal \cite{Wadhia2025}.

\begin{figure}
    \centering
    \includegraphics[width=\columnwidth]{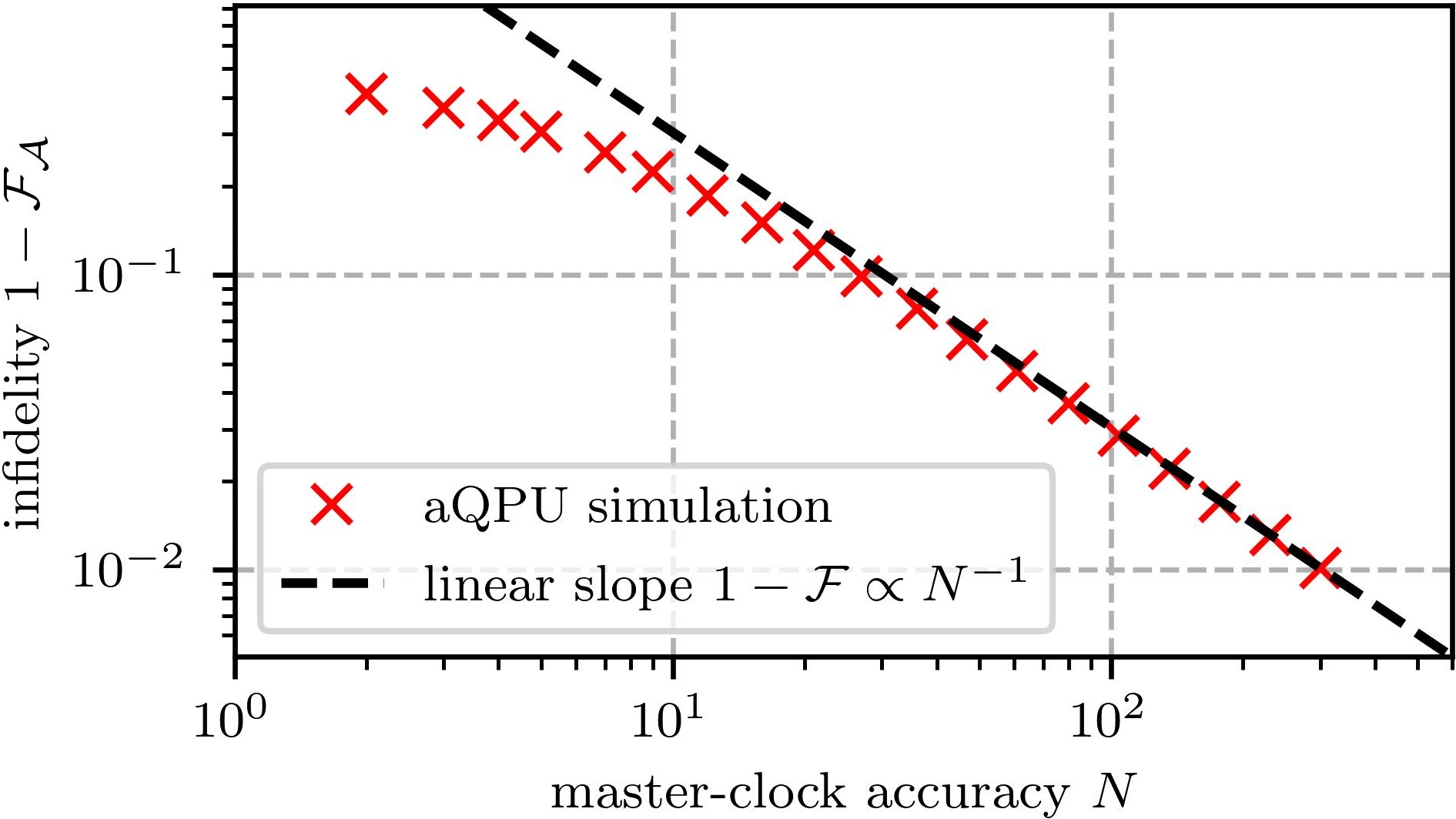}
    \caption{Here we plot the distance of the memory system's state $\rho_M(t)$ at time $t=4\tau$ to the Bell-state $\ket{\Psi_+}=\frac{1}{\sqrt{2}}\left(\ket{00}+\ket{11}\right)$.
    The distance is measured as an infidelity $1-\braket{\Psi_+|\rho_M(t)}{\Psi_+}$ plotted in the $y$-axis and the clock accuracy $N$ on the $x$-axis both in a loglog-plot.
    The aQPU setup is the same as for the plots in Fig.~\ref{fig:evolution} and we model a aQPU performing a Hadamard and CNOT gate on the initial state $\ket{00}.$
    As the clock accuracy $N$ grows, the distance between $\rho_M(t)$ and $\ket{\Psi_+}$ asymptotically drops linearly in $N$, as predicted by Eq.~\eqref{eq:prec}.
    }
    \label{fig:fidelity}
\end{figure}

\paragraph*{Numerical methods.}
Numerically simulating the aQPU using directly the standard Lindblad master-equation approach would be unfeasible because the dimensionality of the matrix superoperator $\mathcal L_\mathrm{aQPU}$ is $(D\times 3\times 3^3\times 4)^4.$
However, one can take advantage of the sparsity of the problem directly, owing the fact that the clockwork evolution is purely classical; hence, instead of working in the space of density matrices for the clock, tick register and punch card, we can use classical probability distributions, effectively improving the performance quadratically.
We implement this using a matrix of dimension $D^2 \times 3^2$ for joining the clock, tick-register and punch card together, detailed on the repository~\cite{Meier2024Git}.
For integration of $e^{t\mathcal L_\mathrm{aQPU}},$ the library \verb|scipy.integrate.solve_ivp| is used.

\section{\label{appendix:generalization_for_backwards_ticking_clocks}Generalization for backwards-ticking clocks}
In Sec.~\ref{sec:termodynamic_cost_of_computation} we provided a preliminary discussion of the thermodynamic cost of running the aQPU taking into account the master clock, and state-preparation but not including the cost of tick generation.
The bound we provided in that setting can be understood as a lower estimate for the entropic cost of the clock.
For a tighter and more universal relationship between entropic cost and clock quality, the cost of tick generation has to be accounted for as well.
Since the tick is a stochastic jump process described by the operator $J$, detailed balance implies the need for a reverse jump process $\overline J$ proportional to $J^\dagger$.
The constant of proportionality between $\overline J$ and $J^\dagger$ is related to the entropy production $\Delta s$ in the thermal baths mediating the jump process and given by~\cite{Seifert2012,Maes2021},
\begin{align}
\label{eq:J_untick}
    \overline J = e^{-\Delta s / 2}J^\dagger,
\end{align}
as already pointed out in Eq.~\eqref{eq:L_ell_detailed_balance} for the internal clockwork dynamics.
Formally, only for divergent $\Delta s\rightarrow\infty,$ this reverse process vanishes completely.
In this appendix, we discuss what implications a finite value for $\Delta s$ has on the aQPU.
Based on the results from Ref.~\cite{Silva2023} on quantum clocks, we show how the case of finite $\Delta s$ reduces to the established equations~\eqref{eq:dot_rho_T} from the main text in the limit of $\Delta s\rightarrow\infty$.

\paragraph*{Reversible clock ticks.}
We start this appendix with a note on the physics of clocks which \textit{can} tick reversibly.
In an everyday picture of clocks, a wall-clock for example, it is most unlikely that the second hand jumps backwards, but it is not impossible.
Thermal fluctuations in the clockwork in general allow for such a process, unless the clock is run at absolute zero temperature.
Nonetheless, if said clock jumps back by one second, for an observer reading the clock, time does not jump backwards by one second. The observer will simply see the clock undergoing a stochastic fluctuation in the direction opposed to its more likely path of evolution.
In a similar way, we can think of the $\overline{J}$ process in the clock description used for this work: if the master clock and the tick register are run at finite temperatures (implying finite $\Delta s$), the tick register can jump backwards.
This does not mean time runs backwards.
In principle, an observer measuring the tick register can detect these jumps by creating a temporal record of the register's state, analogously to how one could see the hand of a wall clock jump back by one second.
Nonetheless, if the interactions between the master clock, instruction register, and memory system are given by the state in which the tick register finds itself in, these reverse jumps affect the computation. As a result, if the clock jumps backwards, it may be that the sequence of instructions from the punch card state $\ket{\mathcal A}_R$ that are carried out is $\cdots a_n\rightarrow a_{n-1} \rightarrow a_n \rightarrow a_{n+1} \cdots$ instead of $\cdots a_n \rightarrow a_{n+1} \cdots$ as one would usually desire.
To quantitatively capture how these different processes affect the computation, we need to discuss the two points:
\begin{itemize}
    \item How does the aQPU memory system evolve in the case both forwards and backwards ticks are possible?
    \item What is the relationship between entropy production of the clock and the reverse tick processes? 
\end{itemize}

\paragraph*{Evolution equations.}
The only change to $\mathcal L_\mathrm{aQPU}$ as defined in~\eqref{eq:L_aQPU} in the main text is that we add a term corresponding to $\overline{J}.$
We remind ourselves that $J$ is defined as,
\begin{align}
    J = \sum_{n\geq 0} J_C \otimes \ketbra{n+1}{n}_{T},
\end{align}
which allows us to determine the reverse process through~\eqref{eq:J_untick} and the operator
\begin{align}
\label{eq:J_untick_detail}
    \overline J = \sum_{n\geq 0} \overline{J}_C \otimes \ketbra{n}{n+1}_{T}.
\end{align}
For the case where there is only the forward ticking process, and the ticks happen with unit probability in the infinite time limit, one can define the tick probability density ${\rm Pr}[T_n=t]$.
For the case where the clock's register can jump both forward and backward, the notion of a tick probability density function needs refinement, and one can for example switch to a different picture.
One immediate generalization considered in~\cite{Silva2023} are tick currents which in our case would be given as
\begin{align}
    p_n(t) = \tr\left[J_C^\dagger J_C \rho_C^{(n-1)}(t)\right],
\end{align}
for the forward tick current of the $n$th tick, the generalization of the expressions in Lemma~\ref{lemma:JumpTickProbDensity}.
One can think of tick currents as a tick rate, i.e., the number density of ticks per unit time, which does not necessarily have to be normalized quantity.
For the reverse process, we would have 
\begin{align}
    \overline{p}_n(t) = \tr\left[\overline{J}_C^\dagger \overline{J}_C \rho_C^{(n+1)}(t)\right],
\end{align}
the reverse current of the $n$th tick.
We can then follow the steps outlined already in Appendix~\ref{sec:aQPU_universality} to derive the evolution equations for the memory system's state $\rho_M^{(n)}(t).$
The only change is that now, there is a new term in $\mathcal L_\mathrm{aQPU}$ given by the dissipator $\mathcal D[\overline{J}]$, and furthermore we have
\begin{align}
    \frac{d}{dt}{\rm Pr}[N(t)=n] = p_n(t) - p_{n+1}(t) + \overline{p}_n(t) + \overline{p}_{n-1}(t),
\end{align}
instead of $d/dt {\rm Pr}[N(t)=n]=p_n(t)-p_{n+1}(t)$ as in Eq.~\eqref{eq:tr_dot_rho_n}.
Combining all terms together, we arrive at the evolution equations for the memory system, where a new term with prefactor $\overline{p}_n(t)$ appears as the contribution from the reverse-ticks.
One consequence is that it is not possible anymore to find a recursion relation as in Prop.~\ref{prop:targetrecursion} because now $\dot\rho_M^{(n)}(t)$ depends not only on $\rho_M^{(n)}(t)$ and $\rho_M^{(n-1)}(t),$ but also on $\rho_M^{(n+1)}(t)$ due to the backwards ticks,
\begin{widetext}
    \begin{align}
        \dot \rho_M^{(n)}(t) = -i\left[H_{M,a_n},\rho_M^{(n)}(t)\right] + \frac{1}{{\rm Pr}[N(t)=n]}\left(p_n(t)\left(\rho_M^{(n-1)}(t)-\rho_M^{(n)}(t)\right) +\overline{p}_n(t)\left(\rho_M^{(n+1)}(t)-\rho_M^{(n)}(t)\right)\right).
    \end{align}    
\end{widetext}
As we take the limit $\Delta s\rightarrow\infty$ of divergent entropy production per unit population undergoing a tick of the clock, the term $\overline{p}_n(t)\rightarrow 0$ vanishes.
For example, in the case where the tick register models a macroscopic memory where for all practical purposes one can assume $\overline{p}_n(t)= 0,$ we recover the previous evolution equations from the main text.
For all the cases, where $\Delta s<\infty$ takes a finite value, the entropy production for the ticking process can be calculated from detailed balance and for a more detailed estimate of the true thermodynamic cost of the aQPU, this contribution has to be taken into account as well.
Similar to the internal clockwork contribution $\Sigma_\mathrm{cw}(t)$ discussed already in Sec.~\ref{sec:termodynamic_cost_of_computation}, the ticks also contribute to the dissipation.
A detailed analysis for the case of reversible ticks is left for future work to analyze.

\providecommand{\noopsort}[1]{}\providecommand{\singleletter}[1]{#1}%

\end{document}